\documentclass[global]{svjour}
\usepackage{graphics}
\usepackage{amsmath}
\usepackage{amssymb}
\usepackage{graphicx}

\begin{document}
\title{Mappings of open quantum systems onto chain representations and Markovian embeddings}
\titlerunning{Chain mappings of open quantum systems}
\author{M.P. Woods\inst{1,}\inst{2}  \and R. Groux\inst{3} \and A.W. Chin\inst{1} \and S.F. Huelga\inst{1} \and and M.B. Plenio\inst{1,}\inst{2}
}                     
\institute{Institute f\"ur Theoretische Physik, Universitatet Ulm, D-89069 Ulm, Germany \and QOLS, Blackett Laboratory, Imperial College London, SW7 2BW, United Kingdom \and Lyc\'ee Polyvalent Rouvi\`ere, Rue Sainte Claire Deville. BP 1205, 83070 Toulon, France\\
\email{mischa.woods05@imperial.ac.uk,\, roland.groux@orange.fr,\, alex.chin@uni-ulm.de,\\ susana.huelga@uni-ulm.de,\, martin.plenio@uni-ulm.de}}

\date{Received: date / Accepted: date}
%

%
\maketitle
\begin{abstract}
We derive a sequence of measures whose corresponding Jacobi matrices
have special properties and a general mapping of an open quantum system
onto 1D semi infinite chains with only nearest neighbour interactions. Then we proceed
to use the sequence of measures and the properties of the Jacobi matrices to
derive an expression for the spectral density describing the open quantum system
when an increasing number of degrees of freedom in the environment have been embedded into the system. Finally, we derive
convergence theorems for these residual spectral densities.
\end{abstract}
\section{Introduction}
\subsection{Background}\label{Background}
Real quantum mechanical systems are never found in complete isolation, but 
invariably coupled to a macroscopically large number of "environmental" degrees 
of freedom, such as those provided by electromagnetic field modes, density 
fluctuations of the surrounding media (phonons) or ensembles of other quantum 
systems, like electronic or nuclear spins \cite{weiss,OQS,rh}. The fact that 
the environment is totally or partially inaccessible to experimental probing 
in these so called "open quantum systems" leads to the appearance of an effectively 
irreversible dynamics for the quantum system's observables, and mediate the 
fundamental processes of energy relaxation, phase decoherence and, possibly, 
the thermalization of the subsystem. 

Accurate numerical or analytical description of general open quantum systems dynamics
appears, prima facie, to be extremely difficult due to the large (often infinite) 
number of bath and system variables which need to be accounted for. When the 
environmental degrees of freedom are modeled as a {\em bath} of harmonic oscillators 
exact path integral solutions can be available but are rarely of practical use 
\cite{SB exp cut,weiss,OQS}. Hence assumptions such as weak system-environment coupling 
and vanishing correlation times of the environment, i.e. the Born-Markov approximation, 
are often invoked to obtain compact and efficiently solvable equations. These approaches 
suffer the drawback that their accuracy is hard to certify and that they become simply 
incorrect in many important situations. Indeed, our increasing ability to observe and 
control quantum systems on ever shorter time and length scales is constantly revealing 
new roles of noise and quantum coherence in important biological and chemical processes \cite{Engel,Mohseni,MartinSusana,Caruso,BioGen} and requires an accurate but efficient
description of the system-environment interaction that go well beyond the Born-Markov 
approximation \cite{BioGen,BioGen2} in order to understand the interaction of intrinsic 
quantum dynamics and environmental noise. In many biological, chemical and solid-state 
systems, deviations from strict Markovianity, which can be explicitly quantified 
\cite{Angel,BreuerNon,jens}, are significant and methods beyond standard perturbative 
expansions are required for their efficient description. A number of techniques have been 
developed to operate in this regime. Those include polaron approaches \cite{polaron}, the 
quasi-adiabatic path-integral (QUAPI) method \cite{quapi}, the hierarchical equation of 
motion approach \cite{aki} and extensions of the quantum state diffusion description to 
non-Markovian regimes \cite{strunz}.

Here we will focus on the exploration of the mathematics of an exact, analytical mapping 
of the standard model of a quantum system interacting with a continuum of harmonic oscillators 
to an equivalent model in which the system couples to one end of a chain nearest-neighbour 
coupled harmonic oscillators, as illustrated in (a) and (b) of Fig. (\ref{fig:myFig}). This mapping has permitted the formulation of an efficient algorithm
for the description of the system-environment coupling for arbitrary spectral densities of 
the environment fluctuations \cite{Alex,Javi}. This new mapping was originally intended 
just as a practical means of implementing t-DMRG which would avoid approximate determination 
of the chain representation using purely numerical, and often unstable, transforms \cite{Bulla}. 
It was quickly realised that the scope of the mapping is much broader. Indeed, the mapping 
itself provides an extremely intuitive and powerful way of analysing universal properties of 
open quantum systems, that is, independent of the numerical method used to simulate the dynamics 
\cite{Alex}. This conclusion was implied also in \cite{Bassano}, where the authors also developed 
a, in principle, rather different chain representation of a harmonic environment using an 
iterative propagator technique \cite{Bassano}.

Both of these theories establish chain representations as a novel and direct way of looking 
at how energy and correlations propagate into the environment in real time.  In the chain 
picture the interactions cause excitations to propagate away from the system, allowing a 
natural, causal understanding of Markovian and non-Markovian dissipation in terms of the 
properties of the chain's couplings and frequencies.  An intuitive account of the physics 
of chain representations, non-Markovian dynamics and irreversibility for the method in 
\cite{Javi,Alex} is given in \cite{bookchap}, and an interpretation of chain parameters 
in terms of time-correlations for the method in \cite{Bassano} can be found in \cite{irene new paper}.

\subsection{This communication}
The main goals of this paper can be split into two groups. Firstly, we aim at developing a general framework for mapping the environment of an open quantum system onto semi infinite 1D chain representations with nearest neighbour interactions where the system only couples to the first element in the chain. In these chain representations, there is a natural and systematic way to "embed" degrees of freedom of the environment into the system (by "embed", we mean to redefine what we call system and environment by including some of the environmental degrees of freedom in the system) (c) of figure (\ref{fig:myFig}). One can make a non-Markovian system-environment interaction more Markovian by embedding some degrees of freedom of the environment into the system, a technique already employed in certain situations in quantum optics \cite{pseudomodes} and that recently has been demonstrated in \cite{Bassano}. What remains unclear however, is to quantify how efficient such procedures can be as well as determining the best way of performing the embedding. We will show that those issues can be efficiently addressed with the formalism presented in this manuscript.

This general formalism allows the comparison between different chain mappings. Thus the second aim is to develop a method for understanding how Markovian such embeddings are by finding explicit analytical formulas for the spectral densities of the embeddings corresponding to the different chain mappings. Furthermore, we also derive universal convergence theorems for the spectral densities corresponding to the embedded systems and give rigourous conditions for when these limiting cases are achieved. This paves the way for a system interacting with a complex environment to be recast by moving the boundary of the system and environment, so that the non-trivial parts of the environment are embedded in the new effective "system" and the homogeneous chain represents the new, and much simpler, "environment" - See (c) in figure (\ref{fig:myFig}) for a pictorial illustration.  The advantage of this is that the residual part of the environment  might be simple enough for some of the approximations mentioned in section \ref{Background} to be applied, enabling us to integrate out these modes and dramatically reduce the number of sites of the chain that have to be accounted for explicitly.
In order to achieve this, we have to first develop new mathematical tools and theorems regarding secondary measures and Jacobi operators, greatly extending and developing the application of orthogonal polynomials that was used in the original chain mapping of \cite{Alex}. These results might be useful also in other areas of mathematics and mathematical physics which are related to the theory of Jacobi operators such as the Toda lattice \cite{Gerard}.

We show that this generalised chain mapping reduces to two known results mentioned in section \ref{Background} under special conditions. For the first known result \cite{Javi,Alex}, the general method developed in this paper gives analytical and non-iterative expressions for the spectral densities corresponding to the embeddings. For the other special case \cite{Bassano}, we derive calculable conditions for when the spectral densities corresponding to the embedding converge, -an aspect not addressed in \cite{Bassano}.
As seen in the examples, we apply this technique to derive exact solutions for the family of Spin-Boson models which will allow us to illustrate how the different embedding methods are related. In addition, the method developed here is also valid when the spectral density of the system-enviroment interaction has a gap in its support region. This is of practical interest as there are open quantum systems (such as photonic crystals \cite{cristal}) that naturally exhibit such a spectral density and hence can only be mapped onto a chain using the method presented in this paper.

The contents of each section is as follows: Section \ref{Secondary measures} is concerned with deriving the necessary mathematical tools for the application to open quantum systems in the subsequent sections.  We start by introducing some elementary results in the field of orthogonal polynomials in section \ref{Introduction to notation and basic tools} (this section also helps to introduce notation). In section \ref{Derivation of the sequence of secondary normalised measures} we focus on deriving a formula which makes explicit a sequence of secondary measures solely in terms of the initial measure and its orthogonal polynomials. We point out that although authors such as Gautschi introduce the concept and definition of secondary measures, here we provide for the first time an analytical closed expression for them in terms of the initial measure and its orthogonal polynomials. This result will turn out to be a vital ingredient in the development of the subsequent chapters.  Moreover, Gautschi states that the general solution we have found is unknown \cite{Gautschi} p16-17.
In section \ref{Derivation of the Jacobi matrix theorem} we study the properties of the 3-term recursion coefficients of the orthogonal polynomials generated from their measures. We then define beta normalised measures which have more general and useful properties. The main result of this section is our new theorem regarding the Jacobi matrix. This new theorem will be used extensively in section \ref{Chain mappings of open quantum systems and Markovian embeddings}.
Section \ref{Chain mappings of open quantum systems and Markovian embeddings} is concerned with chain mappings for open quantum systems and embeddings of environmental degrees of freedom into the system. In section \ref{Generalised bosonic mapping} we develop a general framework for mapping open quantum systems which are linearly coupled to an environment onto a representation where the environment is a semi-infinite chain with nearest neighbour couplings and define two special cases of particular interest, - the particle and phonon mappings.  In section \ref{Connection between the phonon mapping to previous work and the sequence of partial spectral densities}, we investigate the relation between this work and recent work by \cite{Bassano}. In order to do this we make extensive use of the relations developed in section 2. We show how their mapping is a special case of the work presented here and find analytical non iterative solutions to quantities such as the sequence of residual spectral densities. Section \ref{Sequence of partial spectral densities for the particle mapping case} is dedicated to deriving the formulas for the sequence of residual spectral densities for the particle case mapping. In section \ref{convergence section} we develop convergence theorems for the sequence of residual spectral densities. We show rigorously that the sequence converges under certain conditions and give the universal functions the spectral densities converge to for the particle and phonon cases. The conditions for which the sequences converge are stated in terms of the initial spectral density.
In section \ref{examples section} we give explicit analytic examples for the family of spectral densities of the Spin-Boson model for the particle and phonon mapping cases.\\
The main new results of this article are theorems (\ref{1st main theorem}), (\ref{jacobi matrix theorem}), (\ref{generalised mapping theorem}), (\ref{the partial sd generade from a measure in Bassano et all}), (\ref{the partial sd Bassano et all equiv theorem}), (\ref{theorem for the sequence of sd in the particle m case}), (\ref{SD convergence thorem}) and corollaries (\ref{the other main corrollary}) and (\ref{the phonon mapping corollary rof psd}).
\section{Secondary measures}\label{Secondary measures}
\subsection{Introduction to notation and basic tools}\label{Introduction to notation and basic tools}
\begin{definition}\label{meausre defn}
Let us consider a \textit{measure} $d\mu(x)=\bar\mu(x)dx$ with real support intervals $I_1,I_2,I_3,\ldots,I_k$ and $I=[\textup{inf}\{y,\,y\!\in I_1\cup I_2\cup I_3\cup\ldots \cup I_k\}, \textup{sup}\{y,\,y\!\in I_1\cup I_2\cup I_3\cup\ldots \cup I_k\}]=[a,b]$ with $a<b,$ which has finite moments:
\begin{equation}
C_n(d\mu)=\int_a^bx^nd\mu(x) \quad n=0,1,2,\ldots,\;\
\end{equation}
with $\bar \mu(x)\geq0$ over $I$.\\
We call $d\mu$ \textup{gapless} measure if \begin{equation}
I=I_1\cup I_2\cup I_3\cup\ldots \cup I_k
\end{equation}
and \textup{gapped} measure if 
\begin{equation}
I\neq I_1\cup I_2\cup I_3\cup\ldots \cup I_k.
\end{equation}
If not stated otherwise, a measure can be gapped or gapless and $a,b$ finite or infinite.
\end{definition}
\begin{remark}\label{measure detials}
The distinction between $a,b$ finite or infinite is important because as we will see, a few of the results in this paper (most noticeably, theorems \ref{convergence of bounded mesures} and \ref{SD convergence thorem}), are restricted to the finite interval case. We also point out that another important classification of measures is that of \textit{gapped} and \textit{gappless} measures. Gapped measures also have much physical relevance and many of the theorems developed here apply equally well to these cases (see remark (\ref{gaps}) for more details). It is also worth noting that in the case that $I$ is an unbounded interval, one has to be especially cautious when defining a measure, that all the moments are finite. Later on in this article (in proof to theorem \ref{generalised mapping theorem}), we will define measures as functions of spectral densities. Hence, for the case of spectral densities with unbounded support, this spectral density will only define a measure if the corresponding measure has finite moments. In practice, this is not really much of a restriction as the unbounded spectral densities found in the literature do have finite moments (such as in example \ref{The power law spectral densities with exponential cut offff}).\\
As a final remark regarding measures, we note that the above definition is that of a positive measure. Some of the basic theorems and definitions of this section are also valid when the measure is not strictly positive. However, such cases are irrelevant for the new theorems developed here because a spectral density is a non-negative function. For this reason, we deal with only positive measures. 
\end{remark}
\begin{definition}\label{def inner product}
Let $\mathbb{P}$ denote the space of real polynomials. Then, for any $u(x)$ and $v(x)$ $\in\mathbb{P}$ we will define an \textup{inner product} as
\begin{equation}
\langle u,v\rangle_{\bar\mu}=\int_a^b u(x)v(x)d\mu(x).
\end{equation}
\end{definition}
\begin{definition}\label{def of Pn} We call $\{P_n(d\mu;x)\}_{n=0}^\infty$  the set of \textup{real orthonormal polynomials with respect to measure $d\mu$} where each polynomial $P_n$ is of degree  $n$, if they satisfy
\begin{equation}\label{orthognality conditions ofr norm pols eq}
 \langle P_n(d\mu),P_m(d\mu)\rangle_{\bar\mu}=\delta_{nm} \quad n,m=0,1,2,\ldots\ .
\end{equation}
\end{definition}
Similarly,
\begin{definition}\label{def of monic} We call $\{\pi_n(d\mu;x)\}_{n=0}^\infty$ the set of \textup{real monic polynomials with respect to measure $d\mu$} where each polynomial $\pi_n$ is of degree $n$ if they satisfy
\begin{equation}\label{monic poly defn}
\pi_n(d\mu;x)=P_n(d\mu;x)/a_n \quad n=0,1,2,\ldots\,,
\end{equation}
where $a_n=a_n(d\mu)$ is the leading coefficient of $P_n(d\mu;x)$.
\end{definition}
\begin{theorem}
For any measure $d\mu(x)$, there always exists a set of real orthonormal polynomials and real monic polynomials.
\end{theorem}
\begin{proof} See \cite{Gautschi}.\end{proof}
\begin{theorem}\label{monic recursion relation}
The monic polynomials $\pi_n(d\mu;x)$ satisfy the three term recurrence relation
\begin{eqnarray}
\pi_{n+1}(d\mu;x)&=(x-\alpha_n)\pi_n(d\mu;x)-\beta_n\pi_{n-1}(d\mu;x)& \quad n=0,1,2,\ldots\,,\label{3-term recurrece eq}\\
&\pi_{-1}(d\mu;x):=0,&
\end{eqnarray}
where
\begin{eqnarray}
\alpha_n&=&\alpha_n(d\mu)=\frac{\langle x\pi_n(d\mu),\pi_n(d\mu)\rangle_{\bar\mu}}{\langle \pi_n(d\mu),\pi_n(d\mu)\rangle_{\bar\mu}} \quad n=0,1,2,\ldots,\label{the alpha def eq}\\
\beta_n&=&\beta_n(d\mu)=\frac{\langle \pi_n(d\mu),\pi_n(d\mu)\rangle_{\bar\mu}}{\langle \pi_{n-1}(d\mu),\pi_{n-1}(d\mu)\rangle_{\bar\mu}} \quad n=1,2,3,\ldots \;.\label{bet def eq}
\end{eqnarray}
\end{theorem}
\begin{proof} See \cite{Gautschi}.\end{proof}

\begin{definition} We will define $\beta_0(d\mu)$ by
\begin{equation}
\beta_0(d\mu)=\langle \pi_0(d\mu),\pi_0(d\mu)\rangle_{\bar\mu}.\label{the beta 0 def eq}
\end{equation}
\end{definition}
\begin{corollary}\label{cor C0 beta0}
$\beta_0(d\mu)=C_0(d\mu)$.
\end{corollary}
\begin{proof}
We note that from definition (\ref{def of monic}), $\pi_0(d\mu;x)=1$ for all measures $d\mu(x)$. Hence,
\begin{equation}
\beta_0(d\mu)=\langle \pi_0(d\mu),\pi_0(d\mu)\rangle_{\bar\mu}=\int_a^bd\mu(x)=C_0(d\mu).
\end{equation}\qed \end{proof}
\begin{theorem}
When the measure $d\mu$ has bounded support, the $\alpha_n(d\mu)$ and $\beta_n(d\mu)$ coefficients are bounded by
\begin{eqnarray}
&a<\alpha_n(d\mu)<b &\quad n=0,1,2,\ldots\,,\ \\
&0<\beta_n(d\mu)\leq max(a^2,b^2)&\quad n=0,1,2,\ldots\,.
\end{eqnarray}
\end{theorem}
\begin{proof}
See \cite{Gautschi}.
\end{proof}
\begin{theorem}\label{recurrence for monics}
The orthonormal polynomials $P_n(d\mu;x)$ satisfy the three term recurrence relation
\begin{eqnarray}
&t_nP_{n+1}(d\mu;x)=(x-s_n)P_n(d\mu;x)-t_{n-1}P_{n-1}(d\mu;x)\quad n=0,1,2,\ldots\,,&\quad\\
&P_{-1}(d\mu;x):=0,\quad\ P_0(d\mu;x)=1/\sqrt{\beta_0(d\mu)},&
\end{eqnarray}
where
\begin{eqnarray}
s_n=s_n(d\mu)&=&\alpha_n(d\mu)\quad n=0,1,2,\ldots,\\
t_n=t_n(d\mu)&=&\sqrt{\beta_{n+1}(d\mu)}\quad n=0,1,2,\ldots \;.\label{the t beta rel eq}
\end{eqnarray}
\end{theorem}
\begin{proof} See \cite{Gautschi}.\end{proof}
\begin{definition}\label{defn of 2nd poly} We will call $Q_n(d\mu;x)$ the \textup{secondary polynomial}\footnote{also known as \textit{polynomial of the second kind.}} associated with polynomial $P_n(d\mu;x)$ defined by
\begin{equation}\label{2nd poly}
Q_n(d\mu;x)=\int_a^b\frac{P_n(d\mu;t)-P_n(d\mu;x)}{t-x}d\mu(t),\quad n=0,1,2,\ldots
\end{equation}
\end{definition}
\begin{lemma}
The polynomials $Q_n(d\mu;x),$ $n=1,2,3,\ldots$ are real polynomials of degree $n-1$ and $Q_0=0$.
\end{lemma}
\begin{proof} Follows from writing $P_n(d\mu;x)$ in the form $P_n(d\mu;x)=\sum_{q=0}^n k_qx^q$ , using the identity $t^q-x^q=(t-x)\sum_{p=0}^{p=q-1}t^px^{q-1-p},\quad q=1,2,3,\ldots$ to cancel the denominator in Eq. (\ref{2nd poly}) and noting that $C_0(d\mu)>0$.\qed
\end{proof}
\begin{definition}\label{sitjes def}
The \textup{Stieltjes Transformation} of the gapless measure $d\mu(x)=\bar \mu(x)dx$ is defined by\textup{\cite{Berg}}
\begin{equation}
S_{\bar\mu}(z)=\int_a^b\frac{d\mu(x)}{z-x},
\end{equation}
where $z\in \mathbb{C}-[a,b]$.
\end{definition}
This is a function vanishing at infinity and analytic in the whole complex plane
with the interval $[a,b]$ removed. (If $-a = b = +1$, then $S_{\bar \mu}$ is analytic separately
in Im $z > 0$ and Im $z < 0$, the two branches being different in general.)
\begin{theorem}\label{theorem stiejes}
If a gapless measure $d\rho(x)=\rho(x)dx$ has Stieltjes transformation given by
\begin{equation}\label{eq sit rel}
S_{\bar\rho}(z)=z-C_1(d\mu)-\frac{1}{S_{\bar\mu}(z)},
\end{equation}
with $C_0(d\mu)=1,$ then the secondary polynomials $\{Q_n(d\mu;x)\}_{n=1}^{\infty}$ form an orthogonal family for the induced inner product of $d\rho(x)$ and $d\mu(x)$ is also gapless.
\end{theorem}
\begin{proof} See \cite{Sherman}, or \cite{roland} for a direct proof.\end{proof}

\subsection{Derivation of the sequence of secondary normalised measures}\label{Derivation of the sequence of secondary normalised measures}
\begin{definition}\label{def dary measure}
For two gapless measures $d\rho(x)$ and $d\mu(x)$ satisfying Eq. \textup{(\ref{eq sit rel})}, we call  $d\rho(x)$ the \textup{secondary measure} associated with $d\mu(x)$.
\end{definition}
\begin{definition}\label{def 2nd dary measure}
We call the sequence of gapless measures $d\mu_0,$ $d\mu_1,$ $d\mu_2,\ldots$ generated from a gapless measure $d\mu_0$ by
\begin{eqnarray}
S_{\bar \rho_{n+1}}(z)&=&z-C_1(d\mu_n)-\frac{1}{S_{\bar \mu_n}(z)} \quad n=0,1,2,\ldots\,,\label{norm 2nd meausres eqs1}\\
d\mu_n(x)&=&\bar \mu_n(x)dx \quad n=0,1,2,\ldots\,,\label{norm 2nd meausres eqs2}\\
\bar \mu_n(x)&=&\frac{\bar\rho_n(x)}{C_0(d\rho_n)} \quad n=1,2,3,\ldots\,,\label{norm 2nd meausres eqs3}
\end{eqnarray}
the \textup{sequence of normalised secondary measures}, where $C_0(d\mu_0)=1$.
\end{definition}
We note tha all sequences of normalised secondary measures are gapless by definition. This sequence of measures is a slight adaptation from a basic result in the theory of orthogonal polynomials. In the standard version, the sequence of measures are not normalised. The fact that these objects are actually positive measures, is a well known result of the basic theory. See \cite{Sherman} or \cite{roland}.
\begin{corollary}\label{normalised sencond coll}
All measures in a sequence of normalised secondary measures have their zeroth moment equal to unity.
\end{corollary}
\begin{proof} By taking zeroth moment of both sides of Eq. (\ref{norm 2nd meausres eqs3}) we find
\begin{equation}
C_0(d\mu_n)=\frac{C_0(d\rho_n)}{C_0(d\rho_n)}=1 \quad n=1,2,3,\ldots\,.
\end{equation}
We also have that $C_0(d\mu_0)=1$ by definition (\ref{def 2nd dary measure}). \qed \end{proof}
\begin{lemma}
\begin{equation}
C_n(d\rho_{m+1})=C_{n+2}(d\mu_m)-C_1(d\mu_m)C_{n+1}(d\mu_m)-\sum_{s=0}^{n-1}C_s(d\rho_{m+1})C_{n-s}(d\mu_m)\label{Rolands R eq}
\end{equation}
$n,m=0,1,2,\ldots\,.$
\end{lemma}
\begin{proof}
For simplifity we will prove Eq. (\ref{Rolands R eq}) for $m=0$ as the generalisation is trivial. By Taylor expanding $S_{\bar\mu_0}(z)$ and $S_{\bar\rho_1}(z)$ in $x=1/z$ we find
\begin{eqnarray}
S_{\bar\mu_0}(z)&=&\sum_{n=0}^\infty \frac{C_n(d\mu_0)}{z^{n+1}}\quad\text{as}\,z\rightarrow\infty,\label{expansion mu bar}\\
S_{\bar\rho_1}(z)&=&\sum_{n=0}^\infty \frac{C_n(d\rho_1)}{z^{n+1}}\quad\text{as}\,z\rightarrow\infty.\label{expansion rho bar}
\end{eqnarray}
From Eq. (\ref{eq sit rel}) we have
\begin{equation}\label{muti sitges eq}
S_{\bar\rho_1}(z)S_{\bar\mu_0}(z)=\left(z-C_1(d\mu_0)\right)S_{\bar\rho_1}(z)-1.
\end{equation}
Hence substituting Eq. (\ref{expansion mu bar}) and (\ref{expansion rho bar}) into Eq. (\ref{muti sitges eq}) we find
\begin{equation}
\sum_{n,m=0}^\infty C_n(d\mu_0)C_m(d\rho_1)x^{n+m+2}=\left(\frac{1}{x}-C_1(d\mu_0) \right)\sum_{s=0}^\infty C_s(d\mu_0)x^{s+1}-1.
\end{equation}
By comparing terms of the same power in $x$ and taking into account $C_0(d\mu_0)=1$ we deduce
\begin{equation}\label{the intermediate Rol eq}
\sum_{n=0}^m C_n(d\mu_0)C_{m-n}(d\rho_1)=C_{m+2}(d\mu_0)-C_1(d\mu_0)C_{m+1}(d\mu_0)\quad m=0,1,2,\ldots\,.
\end{equation}
By a change of variable in Eq. (\ref{the intermediate Rol eq}) we finally arrive at Eq. (\ref{Rolands R eq}). \qed
\end{proof}
\begin{lemma}\label{lemma continued frac} A sequence of normalised secondary measures $d\mu_0$, $d\mu_1$, $d\mu_2$,\dots , $d\mu_n$, can be written as a continued fraction of the form
\begin{equation}\label{continued frac 1}
S_0(z)=\cfrac{1}{z-C_{1,0}-\cfrac{d_0}{z-C_{1,1}-\cfrac{d_1}{z-C_{1,2}-\cfrac{d_2}{z-\cdots\cfrac{\cdots}{z-C_{1,n-1}-\cfrac{d_{n-1}}{z-C_{1,n}-d_nS_{n+1}(z)}}}}}}
\end{equation}
where we have introduced the shorthand notation $S_n(z):=S_{\bar\mu_n}(z)$, $C_{n,s}:=C_n(d\mu_s), \quad n,s=0,1,2,\ldots\,;$ $d_{n}:=C_{2,n}-C_{1,n}^2, \quad n=1,2,3,\ldots$, and $d_0:=1$.
\end{lemma}
\begin{proof} By evaluating Eq. (\ref{Rolands R eq}) for $n=0$, and taking into account the above definition of $d_n$ we see that $d_{n}=C(d\rho_n),$ $n=1,2,3,\ldots$. Using our new notation, Eq. (\ref{norm 2nd meausres eqs1}) for the sequence of gapless measures reads
\begin{equation}\label{Sn sequence eq}
S_{n+1}(z)=\frac{1}{d_n}[z-C_{1,n}-\frac{1}{S_n(z)}], \quad n=0,1,2,\ldots \;.
\end{equation}
Solving this for $S_n(z)$ followed by repeated substitution gives us Eq. (\ref{continued frac 1}).\qed \end{proof}
\begin{theorem}\label{quatient theorem}
The following relations hold for the  continued fraction Eq. \textup{(\ref{continued frac 1})}
\begin{equation}
S_0(z)=\frac{u_n(-d_nS_{n+1}(z))+u_{n+1}}{v_n(-d_nS_{n+1}(z))+v_{n+1}}, \quad n=0,1,2,\ldots \;,\label{the S_0 quotient eq}
\end{equation}
with relations
\begin{equation}
u_{n+1}=(z-C_{1,n})u_n-d_{n-1}u_{n-1},\quad v_{n+1}=(z-C_{1,n})v_n-d_{n-1}v_{n-1},\label{u_n and v_n eqs}
\end{equation}
and starting values
\begin{equation}
u_0=0,\quad u_1=1,\quad v_0=1,\quad v_1=z-C_{0,1}.\label{start u and v}
\end{equation}
\end{theorem}
\begin{proof} These are elementary results from the theory of continued fractions (e.g. see section 4, connection with continued fractions \cite{baker2}).\end{proof}
\begin{lemma}\label{pre theorem Pade}
\begin{eqnarray}
\Delta_{n+1}&:=&u_{n+1}v_n-v_{n+1}u_n=d_0d_1d_2\ldots d_{n-1}, \quad n=1,2,3, \ldots\, ,\label{delta eq}\\
\Delta_{1}&=&1.
\end{eqnarray}
\end{lemma}
\begin{proof}
Using Eq. (\ref{u_n and v_n eqs}) to substitute for $u_{n+1}$ and $v_{n+1}$ into Eq. (\ref{delta eq}), we find the relation $\Delta_{n+1}=d_{n-1}\Delta_{n}$. Using Eq. (\ref{start u and v}) to verify Eq. (\ref{delta eq}) for the starting values, Eq. (\ref{delta eq}) follows by induction. \qed
\end{proof}
\begin{definition}\label{defn pade approx}
A \textup{Pad\'e Approximant} for a function $g$ of type $q/p$ in the neighbourhood of $0$ is a rational fraction
\begin{equation}
F(z)=\frac{Q(z)}{P(z)},
\end{equation}
with degree of $Q\leq q$, degree of $P\leq p$ and $g(z)-\frac{Q(z)}{P(z)}$ of order $\mathcal{O}(z^{p+q+1})$ in the neighbourhood of $0$.
\end{definition} For more details, see \cite{baker2}.
\begin{theorem}\label{theorem Pade approx}
$F_n(z)=\frac{u_{n+1}(z)}{v_{n+1}(z)}$ is a Pad\'e Approximant for $S_0(z)$ of type $n/(n+1), \quad n= 0,1,2,\ldots$ .
\end{theorem}
\begin{proof} Using theorem (\ref{quatient theorem}), we can write $S_0(z)-\frac{u_{n+1}(z)}{v_{n+1}(z)}$ as
\begin{equation}
S_0(z)-\frac{u_{n+1}(z)}{v_{n+1}(z)}=\frac{\Delta_{n+1}d_n S_{n+1}(z)}{v_{n+1}(z)(v_{n+1}(z)-d_nv_n(z)S_{n+1}(z))}, \quad n= 0,1,2,\ldots\,. \label{the frac eq intermed at big}
\end{equation}
Through lemma (\ref{pre theorem Pade}) we see that $\Delta_{n+1}$ is independent of $z$. By Taylor expanding $S_{n+1}(z)$ defined in definition (\ref{sitjes def}) about $x=1/z$, and remembering that $C_0(d\mu_n)=1 \quad n=1,2,3\ldots\,\ $, we find using Eq. (\ref{the frac eq intermed at big}) that
\begin{equation}
S_{n+1}(z)=\frac{1}{z}+\mathcal{O}(\frac{1}{z^2})\quad\ \text{as} \:z \rightarrow \infty\quad n=0,1,2,\ldots\,.
\end{equation}
By induction, we see that $v_{n+1}$ and $u_{n+1}, \quad n=0,1,2,\ldots$ given by Eq. $(\ref{u_n and v_n eqs})$ and $(\ref{start u and v})$, are degree $n+1$ and $n$ polynomials in $z$ respectively, both with leading coefficients equal to unity. Hence we conclude that
\begin{equation}
S_0(z)-\frac{u_{n+1}(z)}{v_{n+1}(z)}=\frac{\Delta_{n+1}d_n}{z^{2n+3}}=\frac{d_0d_1\ldots d_n}{z^{2n+3}}\quad\ \text{as} \:z \rightarrow \infty\quad n=0,1,2,\ldots\,.\label{dod1...}
\end{equation}
Thus by definition (\ref{defn pade approx}) we conclude the proof.\qed \end{proof}
\begin{lemma}\label{lemma un and vn}
$u_{n}(z)=\lambda_{n}Q_n(d\mu_0;z)$ and $v_{n}(z)=\lambda_{n}P_{n}(d\mu_0;z)\quad n=0,1,2,\ldots\,$, with $\lambda_n=1/a_n$ were $a_n$ is defined in definition \textup{(\ref{def of monic})}.
\end{lemma}
\begin{proof} From section 5.3.: moment problems and orthogonal polynomials (p213-220) of \cite{baker2}, and theorem (\ref{theorem Pade approx}) it follows that
\begin{equation}
\frac{u_{n+1}(z)}{v_{n+1}(z)}=\frac{Q_{n+1}(d\mu_0;z)}{P_{n+1}(d\mu_0;z)}\quad n=0,1,2,\ldots\,.\label{frac q/p}
\end{equation}
By observing the starting values, we also have that
\begin{equation}
\frac{u_{0}(z)}{v_{0}(z)}=\frac{Q_{0}(d\mu_0;z)}{P_{0}(d\mu_0;z)}.
\end{equation}
Hence
\begin{equation}
u_{n}(z)=\lambda_{n}Q_n(d\mu_0;z)\quad\ \text{and} \quad v_{n}(z)=\lambda_{n}P_{n}(d\mu_0;z)\quad n=0,1,2,\ldots\,.
\end{equation}
Given that $u_n$ and $v_n$ have leading coefficients
equal to 1, we must have $\lambda_n=1/a_n$\quad n=0,1,2,\ldots\,.\qed
\end{proof}
\begin{theorem}\label{the dn beta theorem}
$d_n=a_n^2/a_{n+1}^2=\beta_{n+1}(d\mu_0)\quad n=0,1,2,\ldots\,.$
\end{theorem}
\begin{proof} Proceeding in the same way as in page (18) of \cite{Gautschi}, we have
\begin{equation}
S_0(z)-\frac{Q_n(d\mu_0;z)}{P_n(d\mu_0;z)}=\frac{\gamma_n}{z^{2n+1}}\quad\ \text{as} \:z \rightarrow \infty\quad n=0,1,2,\ldots\,.\label{stitjes tranform id}
\end{equation}
where
\begin{equation}\label{the gamma eq}
\gamma_n=\frac{1}{a_n}\int_a^bx^nP_n(d\mu_0;x)d\mu_0(x)=\frac{\gamma_0a_0^2}{a_n^2}\frac{\langle P_n(d\mu_0),P_n(d\mu_0)\rangle_{\bar\mu}}{\langle P_0(d\mu_0),P_0(d\mu_0)\rangle_{\bar\mu}}\quad n=0,1,2,\ldots\,.
\end{equation}
Noting that $P_0(d\mu_0;x)/a_0=\pi_0(d\mu_0;x)=1$ and that $C_0(d\mu_0)=1$, Eq. (\ref{the gamma eq}) tells us $\gamma_0=1$.
Comparing Eq. (\ref{stitjes tranform id}) with Eq. (\ref{dod1...}) and (\ref{frac q/p}), we deduce that
\begin{equation}\label{the d0d1d2...equate eq}
d_0d_1\ldots d_{n-1}=\frac{a_0^2}{a_n^2}\frac{\langle P_n(d\mu_0),P_n(d\mu_0)\rangle_{\bar\mu}}{\langle P_0(d\mu_0),P_0(d\mu_0)\rangle_{\bar\mu}}\quad n=1,2,3,\ldots\,.
\end{equation}
By induction it follows
\begin{equation}\label{the d_n equation}
d_n=\frac{a_n^2}{a_{n+1}^2}\frac{\langle P_{n+1}(d\mu_0),P_{n+1}(d\mu_0)\rangle_{\bar\mu}}{\langle P_n(d\mu_0),P_n(d\mu_0)\rangle_{\bar\mu}}\quad n=1,2,3,\ldots\,.
\end{equation}
Due to definition (\ref{def of Pn}), we see that $\langle P_n(d\mu_0),P_n(d\mu_0)\rangle_{\bar\mu}=1\quad n=0,1,2,\ldots\,,$ hence
\begin{equation}
d_n=\frac{a_n^2}{a_{n+1}^2}\quad n=1,2,3,\ldots\,.
\end{equation}
From definition (\ref{def of monic}), we see that Eq. (\ref{the d_n equation}) can be written in the form
\begin{equation}
d_n=\frac{\langle \pi_{n+1}(d\mu_0),\pi_{n+1}(d\mu_0)\rangle_{\bar\mu}}{\langle \pi_n(d\mu_0),\pi_n(d\mu_0)\rangle_{\bar\mu}}\quad n=1,2,3,\ldots\,.
\end{equation}
Hence, from definition (\ref{bet def eq}) we conclude,
\begin{equation}
d_n=\beta_{n+1}(d\mu_0)\quad n=1,2,3,\ldots\,.
\end{equation}
For $n=1$, Eq. (\ref{the d0d1d2...equate eq}) gives us
\begin{equation}
d_0=\frac{\langle \pi_{1}(d\mu_0),\pi_{1}(d\mu_0)\rangle_{\bar\mu}}{\langle \pi_0(d\mu_0),\pi_0(d\mu_0)\rangle_{\bar\mu}}=\frac{a_0^2}{a_{1}^2}=\beta_1(d\mu_0).
\end{equation}\qed
\end{proof}
\begin{theorem}\label{S-P inv theorem}
Gapless measures $d\mu(x)=\bar\mu(x)dx$ can be calculated from their Stieltjes transform by
\begin{equation}
\bar\mu(x)=\frac{1}{2\pi i}\lim_{\epsilon\rightarrow0^+}\left[S_{\bar\mu}(x-i\epsilon)-S_{\bar\mu}(x+i\epsilon)\right].
\end{equation}
\end{theorem}
\begin{proof} This result is known as the \textit{Stieltjes-Perron inversion formula}. See \cite{S-P} for more details or Example 2.50: Stieltjes-Perron inversion formula \cite{Gautschi} for an application and more references.
\end{proof}
\begin{definition}\label{def reducer}
We call $\varphi(d\mu;x)$ the \textup{reducer} of gapless measure $d\mu( x)$. It is given by
\begin{equation}
\varphi(d\mu;x)=\lim_{\epsilon\rightarrow0^+}\left[S_{\bar\mu}(x-i\epsilon)+S_{\bar\mu}(x+i\epsilon)\right].
\end{equation}
\end{definition}
See section \ref{Methods for calculating the reducer} for methods for calculating the reducer.
The reducer allows us to write an explicit expression for the secondary measure associated with $\bar\mu(x)$ as follows.
\begin{theorem}\label{Roland initial theorem}
For a gapless measure $d\mu(x)$ with secondary measure $d\rho(x)$, we have
\begin{equation}\label{eq for rho in temrs of mu}
\bar\rho(x)=\frac{\bar\mu(x)}{\frac{\varphi^2(d\mu;x)}{4}+\pi^2\bar\mu^2(x)}
\end{equation}
\end{theorem}
\begin{proof} See \cite{roland}.
\end{proof}
\begin{definition}\label{Zn defn}
We define the functions $Z_n(x)\epsilon\,\mathbb{C}\quad n=0,1,2,\ldots\ $ as
\begin{equation}\label{1st Z n equantion}
Z_n(x)=\frac{\varphi_n(x)}{2}+i\pi\bar\mu_n(x)\quad n=0,1,2,\ldots\,,
\end{equation}
where
$\varphi_n(x):=\varphi(d\mu_n;x)$.
\end{definition}
\begin{lemma}
The following recursion relation hold for $Z_{n+1}(x)$
\begin{equation}\label{2nd Z n euqation}
Z_{n+1}(x)=\frac{1}{d_n}[x-C_{1,n}-\frac{1}{Z_n(x)}], \quad n=0,1,2,\ldots \;.
\end{equation}
\end{lemma}
\begin{proof} We start by finding a relation between the Stieltjes transformation of a gapless measure $d\mu(x)$, and the  reducer of its associated secondary measure $d\rho(x)$:
By definition, we have
\begin{equation}
\varphi(d\rho; x)=\lim_{\epsilon\rightarrow0^+}\left[S_{\bar\rho}(x-i\epsilon)+S_{\bar\rho}(x+i\epsilon)\right].
\end{equation}
Using Eq. (\ref{eq sit rel}), we find
\begin{equation}\label{intermediate reducer eq}
\varphi(d\rho; x)=2\big[x-C_1(d\mu)\big]-\lim_{\epsilon\rightarrow0^+}\frac{S_{\bar\mu}(x-i\epsilon)+S_{\bar\mu}(x+i\epsilon)}{S_{\bar\mu}(x-i\epsilon)S_{\bar\mu}(x+i\epsilon)}.
\end{equation}
Now using theorem (\ref{S-P inv theorem}) and definition (\ref{def reducer}), we find that
\begin{equation}
\lim_{\epsilon\rightarrow0^+}\frac{S_{\bar\mu}(x-i\epsilon)+S_{\bar\mu}(x+i\epsilon)}{S_{\bar\mu}(x-i\epsilon)S_{\bar\mu}(x+i\epsilon)}=\frac{\varphi(d\mu;x)}{\frac{\varphi^2(d\mu;x)}{4}+\pi^2\bar\mu^2(x)}.
\end{equation}
Hence from Eq. (\ref{intermediate reducer eq}) we arrive at
\begin{equation}\label{intermed reducer reducer}
\varphi(d\rho; x)=2\big[x-C_1(d\mu)\big]-\frac{\varphi(d\mu;x)}{\frac{\varphi^2(d\mu;x)}{4}+\pi^2\bar\mu^2(x)}.
\end{equation}
Using the definition of a sequence of normalised secondary measures, definition (\ref{def 2nd dary measure}) and the definition of $d_n$ in Lemma (\ref{lemma continued frac}), from Eq. (\ref{intermed reducer reducer}) we find
\begin{equation}\label{reducer n}
\varphi_{n+1}(x)=\frac{1}{d_n}\left[2\big[x-C_{1,n}\big]-\frac{\varphi_n(x)}{\frac{\varphi^2_n(x)}{4}+\pi^2\bar\mu^2_n(x)}\right], \quad n=0,1,2,\ldots \;.
\end{equation}
Similarly, we can also write theorem (\ref{Roland initial theorem}) for our sequence of normalised secondary measures
using definition (\ref{def 2nd dary measure}) and $d_n$. We find
\begin{equation}\label{1st mu n equation}
\bar\mu_{n+1}(x)=\frac{1}{d_n}\left[\frac{\bar\mu_n(x)}{\frac{\varphi^2_n(x)}{4}+\pi^2\bar\mu^2_n(x)}\right], \quad n=0,1,2,\ldots \;.
\end{equation}
If we write Eq. (\ref{1st Z n equantion}) for $Z_{n+1}(x)$ and then substitute Eq. (\ref{reducer n}) and Eq. (\ref{1st mu n equation}) into the RHS we arrive at Eq. (\ref{2nd Z n euqation}).\qed \end{proof}
\begin{theorem}\label{quatient theorem for Z n}
The following relations hold between $Z_0(x)$ and $Z_{n+1}(x)$
\begin{equation}
Z_0(z)=\frac{u_n(-d_nZ_{n+1}(z))+u_{n+1}}{v_n(-d_nZ_{n+1}(z))+v_{n+1}}, \quad n=0,1,2,\ldots \;,\label{the Z_0 quotient eq}
\end{equation}
were $u_{n}(z)=\lambda_{n}Q_n(d\mu_0;z)$ and $v_{n}(z)=\lambda_{n}P_{n}(d\mu_0;z)\quad n=0,1,2,\ldots\,$, with $\lambda_n=1/a_n$, $a_n$ defined in definition \textup{(\ref{def of monic})}.
\end{theorem}
\begin{proof} By comparing Eq. (\ref{2nd Z n euqation}) with eq. (\ref{Sn sequence eq}), we note that $Z_n(x)$ satisfies the same recursion relation as $S_n(x)$. Hence  theorem (\ref{the S_0 quotient eq})  readily applies to Eq. (\ref{2nd Z n euqation}) if we exchange  $S_{n+1}(x)$ with $Z_{n+1}(x)$ and $S_0(x)$ with $Z_0(x)$. The relations between $u_n$, $v_n$ and $Q_n$, $P_n$ are proven in lemma (\ref{lemma un and vn}). \qed
\end{proof}
We are now ready to state our first main theorem:
\begin{theorem}\label{1st main theorem}
A sequence of normalised secondary measures starting from $d\mu_0$\textup{: }$d\mu_0$, $d\mu_1$, $d\mu_2$,\dots , $d\mu_m$,\dots can be generated from the first measure in the sequence $d\mu_0$ by the formula
\begin{equation}\label{the mun sequence eq}
\bar\mu_n(x)=\frac{1}{t_{n-1}^2(d\mu_0)}\frac{\bar\mu_0(x)}{\left(P_{n-1}(d\mu_0;x)\frac{\varphi(d\mu_0;x)}{2}-Q_{n-1}(d\mu_0;x)\right)^2+\pi^2\bar\mu_0^2( x)P_{n-1}^2(d\mu_0;x)}
\end{equation}
$n=1,2,3,\ldots\,,$
\end{theorem}
where the $t_n$ coefficients are defined in theorem (\ref{recurrence for monics}).
\begin{proof} After solving Eq. (\ref{the Z_0 quotient eq}) for $Z_n(x)$, we find
\begin{equation}\label{intermedate z(x) eq}
Z_n(x)=\frac{a_{n-1}}{a_nd_{n-1}}\frac{Z_0(x)P_n(d\mu_0;x)-Q_{n}(d\mu_0;x)}{ Z_0(x)P_{n-1}(d\mu_0;x)-Q_{n-1}(d\mu_0;x)}\quad n=1,2,3,\ldots\,.
\end{equation}
From theorem (\ref{the dn beta theorem}) we see that
\begin{equation}\label{an quatioen in d eq}
a_{n-1}/a_n=\kappa_n\sqrt{d_{n-1}}\quad n=1,2,3,\ldots\,,
\end{equation}
where $\kappa_n$ is the sign of $a_{n-1}/a_n$.
After using this relation to simplify the $a_{n-1}/a_nd_{n-1}$ coefficient in Eq. (\ref{intermedate z(x) eq}) and substituting for $Z_0(x)$ using definition (\ref{Zn defn}), we take real and imaginary parts to achieve
\begin{equation}\label{the mun eq not finished}
\bar\mu_n(x)=\frac{\kappa_n}{\sqrt{d_{n-1}}}\frac{\bar\mu_0(x)\left[P_{n-1}(d\mu_0;x)Q_n(d\mu_0;x)-P_n(d\mu_0;x)Q_{n-1}(d\mu_0;x) \right]}{\left(P_{n-1}(d\mu_0;x)\frac{\varphi(d\mu_0;x)}{2}-Q_{n-1}(d\mu_0;x)\right)^2+\pi^2\bar\mu_0^2( x)P_{n-1}^2(d\mu_0;x)}
\end{equation}
$n=1,2,3,\ldots\,.$ Using the identities from lemmas (\ref{pre theorem Pade}) and (\ref{lemma un and vn}), we note that \begin{equation}\label{1st im 1}
P_{n}(d\mu_0;x)Q_{n+1}(d\mu_0;x)-P_{n+1}(d\mu_0;x)Q_{n}(d\mu_0;x) =d_0d_1\ldots d_{n-1}/\lambda_n\lambda_{n+1}
\end{equation}
$n=1,2,3,\ldots\;.$ Using theorem (\ref{the dn beta theorem}) and lemma (\ref{lemma un and vn}) to write the RHS in terms of the $a_n$'s, we find
\begin{equation}\label{1st im 2}
d_0d_1\ldots d_{n-1}/\lambda_n\lambda_{n+1}=a_0^2a_{n+1}/a_n=a_0^2/\kappa_{n+1}\sqrt{d_{n}}\quad n=0,1,2,\ldots\;,
\end{equation}
where in the last line we have used Eq. (\ref{an quatioen in d eq}). We now note that by definition (\ref{def of monic}), $P_0(d\mu_0;x)=a_0$. We also have by definition (\ref{def of Pn}), that $\langle P_0(d\mu_0;x),P_0(d\mu_0;x)\rangle_{\bar\mu_0}$ $=1$. Hence
\begin{equation}\label{1st im 3}
1=\langle P_0(d\mu_0;x),P_0(d\mu_0;x)\rangle_{\bar\mu_0}=a_0^2C_0(d\mu_0)=a^2_0.
\end{equation}
Hence from Eq. (\ref{1st im 1}), (\ref{1st im 2}), (\ref{1st im 3}), we find that we can write Eq. (\ref{the mun eq not finished}) as
\begin{equation}\label{the mu n eq}
\bar\mu_n(x)=\frac{1}{d_{n-1}}\frac{\bar\mu_0(x)}{\left(P_{n-1}(d\mu_0;x)\frac{\varphi(d\mu_0;x)}{2}-Q_{n-1}(d\mu_0;x)\right)^2+\pi^2\bar\mu_0^2( x)P_{n-1}^2(d\mu_0;x)}
\end{equation}
$n=1,2,3,\ldots\,.$ Finally,  with the relations from theorems (\ref{the dn beta theorem}) and (\ref{recurrence for monics}) we find $d_{n-1}=t^2_{n-1}(d\mu_0)$ $\quad n=1,2,3,\ldots\,.$\qed
\end{proof}
\subsection{Derivation of the Jacobi matrix theorem}\label{Derivation of the Jacobi matrix theorem}

\begin{lemma}\label{the shift property of alpha beta}
The $\alpha_n(d\mu)$ and $\beta_n(d\mu)$ coefficients defined in theorem \textup{(\ref{monic recursion relation})} are invariant under a change of scale of the measure $d\mu(x)$, while the change in $\beta_0(d\mu)$ scales linearly:
\begin{eqnarray}
\alpha_n(Cd\mu)&=&\alpha_n(d\mu)\quad n=0,1,2,\ldots\label{the alpha shift eq}\\
\beta_n(Cd\mu)&=&\beta_n(d\mu)\quad n=1,2,3,\ldots\\
\beta_0(Cd\mu)&=&C\beta_0(d\mu)\label{the beta 0 shift eq}
\end{eqnarray}
where $C>0$.
\end{lemma}
\begin{proof} First we will show that $\pi_n(Cd\mu;x)=\pi_n(d\mu;x)\quad n=0,1,2,\ldots$\,. \\ From definition (\ref{monic poly defn}), we have that
\begin{equation}
1=\langle P_n(d\mu),P_n(d\mu)\rangle_{\bar\mu} \quad n=0,1,2,\ldots\,.
\end{equation}
By multiplying and dividing by $C$ we find
\begin{equation}
1=\langle \frac{P_n(d\mu)}{\sqrt{C}},\frac{P_n(d\mu)}{\sqrt{C}}\rangle_{\bar{c\mu}} \quad n=0,1,2,\ldots\,.
\end{equation}
From definition (\ref{monic poly defn}), we conclude
\begin{equation}\label{scaled Pn equu}
P_n(Cd\mu;x)=\frac{P_n(d\mu;x)}{\sqrt{C}}\quad n=0,1,2,\ldots\,.
\end{equation}
By observing definition (\ref{def of monic}) and Eq. (\ref{scaled Pn equu}), we conclude
\begin{equation}\label{a n new old}
a_n(Cd\mu)=a_n(d\mu)/\sqrt{C} \quad n=0,1,2,\ldots\,.
\end{equation}
Hence, from Eq. (\ref{scaled Pn equu}), (\ref{a n new old}) and definition (\ref{monic poly defn}) it follows
\begin{equation}\label{rescaled inv pi}
\pi_n(Cd\mu;x)=P_n(Cd\mu;x)/a_n(Cd\mu)=P_n(d\mu;x)/a_n(d\mu)=\pi_n(d\mu;x)
\end{equation}
$n=0,1,2,\ldots\,.$ Now we can see how the $\alpha_n$ and $\beta_n$ coefficients change:\\
Using definition (\ref{monic recursion relation}) and Eq. (\ref{rescaled inv pi}) we have
\begin{eqnarray}
\alpha_n(Cd\mu)&=&\frac{\langle x\pi_n(Cd\mu),\pi_n(Cd\mu)\rangle_{\bar{c\mu}}}{\langle \pi_n(Cd\mu),\pi_n(Cd\mu)\rangle_{\bar{c\mu}}}\\
&=&\frac{\langle x\pi_n(d\mu),\pi_n(d\mu)\rangle_{\bar{\mu}}}{\langle \pi_n(d\mu),\pi_n(d\mu)\rangle_{\bar{\mu}}}\\
&=&\alpha_n(d\mu),\quad n=0,1,2,\ldots\,.
\end{eqnarray}
Similarly, we find
\begin{equation}
\beta_n(Cd\mu)=\beta_n(d\mu),\quad n=1,2,3,\ldots\,.
\end{equation}
In the case of $\beta_0$ we have
\begin{equation}
\beta_0(Cd\mu)=\langle \pi_0(Cd\mu),\pi_0(Cd\mu)\rangle_{\bar{c\mu}}=C\beta_0(d\mu).
\end{equation}\qed
\end{proof}
\begin{theorem}\label{second 1st meausre theorem}
If $d\rho(x)$ is the secondary measure associated with $d\mu(x)$, then
\begin{eqnarray}
\alpha_{n+1}(d\mu)&=&\alpha_n(d\rho)\quad n=0,1,2,\ldots\\
\beta_{n+1}(d\mu)&=&\beta_n(d\rho)\quad n=1,2,3,\ldots
\end{eqnarray}
\end{theorem}
\begin{proof} See \cite{Gautschi}, theorem 1.36 (page 16).
\end{proof}
\begin{lemma}\label{the norm sequence shif theorem v2}
A sequence of normalised secondary measures $d\mu_0(x),$ $d\mu_1(x),$ $d\mu_2(x),$ $\ldots$, satisfy
\begin{eqnarray}
\alpha_{m+n}(d\mu_0)&=&\alpha_n(d\mu_m)\quad n,m=0,1,2,\ldots\,,\label{alpah n and 0 eq}\\
\beta_{m+n}(d\mu_0)&=&\beta_n(d\mu_m)\quad n=1,2,3,\ldots\,,\,m=0,1,2,\ldots\,.\label{beta n and 0 eq}
\end{eqnarray}
\end{lemma}
\begin{proof} Using lemma (\ref{the shift property of alpha beta}) and definition (\ref{def 2nd dary measure}) we find
\begin{eqnarray}
\alpha_{n}(d\rho_m)&=&\alpha_{n}(d\rho_m/C_0(d\rho_m))=\alpha_n(d\mu_m)\quad n=0,1,2,\ldots\,m=1,2,3\ldots\,,\\
\beta_{n}(d\rho_m)&=&\beta_{n}(d\rho_m/C_0(d\rho_m))=\beta_n(d\mu_m)\quad n=1,2,3,\ldots\,m=1,2,3\ldots\,.
\end{eqnarray}
Hence taking into account theorem (\ref{second 1st meausre theorem}) we find
\begin{eqnarray}
\alpha_{n+1}(d\mu_p)&=&\alpha_n(d\mu_{p+1})\quad n,p=0,1,2,\ldots\,,\label{alpha n1 and m1}\\
\beta_{n+1}(d\mu_p)&=&\beta_n(d\mu_{p+1})\quad n=1,2,3,\ldots\,,\, p=0,1,2,\ldots\,\label{beta n1 and m1}.
\end{eqnarray}
We will now proceed to prove Eq. (\ref{alpah n and 0 eq}) by construction: by evaluating Eq. (\ref{alpha n1 and m1}), for $(n,p)$ at $(s-1,m)$, $(s-2,m+1)$, $(s-3,m+2)$,\dots ,$(0,m+s-1)$ for any $m\geq 0,$ $s\geq 1$ we have the following sequence of equations
\begin{eqnarray}
\alpha_s(d\mu_m)&=&\alpha_{s-1}(d\mu_{m+1})\label{1st sequence eq}\\
\alpha_{s-1}(d\mu_{m+1})&=&\alpha_{s-2}(d\mu_{m+2})\label{2st sequence eq}\\
\alpha_{s-2}(d\mu_{m+2})&=&\alpha_{s-3}(d\mu_{m+3})\\
&\vdots& \nonumber\\
\alpha_{1}(d\mu_{m+s-1})&=&\alpha_{0}(d\mu_{m+s}).\label{last sequence eq}
\end{eqnarray}
Thus by repeated substitution we arrive at
\begin{equation}
\alpha_s(d\mu_m)=\alpha_{0}(d\mu_{m+s}),\quad m=0,1,2,\ldots\,,s=1,2,3,\ldots\,.\label{intermediate Mar compressed proof}
\end{equation}
We note that this equation is also valid for $s=0$. If we make the change of variable $s= p+n,$ $m=0$ in Eq. (\ref{intermediate Mar compressed proof}) followed by relabeling indices, we find
\begin{equation}\label{the intermediate 2nd apla shift proof second eq}
\alpha_{m+s}(d\mu_0)=\alpha_0(d\mu_{m+s})\quad m,s=0,1,2,\ldots\,.
\end{equation}
Hence from Eq. (\ref{intermediate Mar compressed proof}) and (\ref{the intermediate 2nd apla shift proof second eq}), we arrive at (\ref{alpah n and 0 eq}). Similarly, we prove Eq. (\ref{beta n and 0 eq}).\qed
\end{proof}

\begin{definition}\label{nu sequence defn}
We call the sequence of gapless measures $d\nu_0(x)$, $d\nu_1(x)$, $d\nu_2(x)$, $d\nu_3(x)$, $\dots$ \textup{beta normalised measures}, where $d\nu_0(x)$ defines the sequence
\begin{equation}\label{the nu n eq}
\bar\nu_n(x)=\frac{\bar\nu_0(x)}{\left(P_{n-1}(d\nu_0;x)\frac{\varphi(d\nu_0;x)}{2}-Q_{n-1}(d\nu_0;x)\right)^2+\pi^2\bar\nu_0^2( x)P_{n-1}^2(d\nu_0;x)}
\end{equation}
$n=1,2,3,\ldots\,,$ and $d\nu_n(x)=\bar\nu_n(x)dx\quad n=0,1,2,\ldots\,.$
\end{definition}
\begin{lemma}\label{beta measures proportional to norm measures lemma} For every sequence of beta normalised measures $\{d\nu_n\}_{n=0}^{n=\infty}$, there always exists a sequence  of normalised secondary measures $\{d\mu_n\}_{n=0}^{n=\infty}$ such that
\begin{equation}\label{nu lemma eq}
\bar\nu_n(x)=\beta_n(d\nu_0)\bar\mu_n(x)\quad n=0,1,2,\ldots\,.
\end{equation}
\end{lemma}
\begin{proof}
For $n=0$ in Eq. (\ref{nu lemma eq}) and taking into account corollary (\ref{cor C0 beta0}), we have
\begin{equation}\label{nu 0 eq}
\bar\nu_0(x)=C_0(d\nu_0)\bar\mu_0(x).
\end{equation}
The only additional constraint on $d\mu_0$ as compared with any measure $d\mu$, is that $C_0(d\mu_0)=1$. By calculating the zeroth moment of both sides, we see that relation Eq. (\ref{nu 0 eq}) satisfies this additional constraint. Now we will proceed by finding an expression for the sequence of normalised measures generated from $d\mu_0$ in terms of $d\nu_0$. For this we need to find how the quantities $P_n(d\mu_0;x),$ $Q_n(d\mu_0;x),$ $\beta_{n+1}(d\mu_0),$ $n=0,1,2,\ldots$ and $\varphi(d\mu_0;x)$ can be written in terms of  $P_n(d\nu_0;x),$ $ Q_n(d\nu_0;x),$ $\beta_{n+1}(d\nu_0),$ $n=0,1,2,\ldots$ and $\varphi(d\nu_0;x)$ respectively. Using relation Eq. (\ref{nu 0 eq}) and definition (\ref{def of Pn}), we find
\begin{equation}\label{for sub beta proof1}
P_n(d\mu_0;x)=\sqrt{C_0(d\nu_0)}P_n(d\nu_0;x)\quad n=0,1,2,\ldots\,.
\end{equation}
Using this relation, and definitions (\ref{defn of 2nd poly}) and (\ref{def reducer}), we find
\begin{eqnarray}
Q_n(d\mu_0;x)&=&\frac{Q_n(d\nu_0;x)}{\sqrt{C_0(d\nu_0)}}\quad n=0,1,2,\ldots\,,\label{for sub beta proof2}\\
\varphi(d\mu_0;x)&=&\frac{\varphi(d\nu_0;x)}{C_0(d\nu_0)}.\label{for sub beta proof3}
\end{eqnarray}
Lemma (\ref{the shift property of alpha beta}) tells us
\begin{equation}\label{mu nu beta eq}
\beta_n(d\mu_0)=\beta_n(d\nu_0)\quad n=1,2,3,\ldots\,.
\end{equation}
Hence using relation Eq. (\ref{the t beta rel eq}), we find
\begin{equation}
t^2_{n-1}(d\mu_0)=t^2_{n-1}(d\nu_0)=\beta_n(d\nu_0)\quad n=1,2,3,\ldots\,.\label{for sub beta proof4}
\end{equation}
Now substituting Eq. (\ref{for sub beta proof1}), (\ref{for sub beta proof2}), (\ref{for sub beta proof3}), and (\ref{for sub beta proof4}) into Eq. (\ref{the mun sequence eq}),
\begin{equation}
\bar\mu_n(x)=\frac{1}{\beta_n(d\nu_0)}\frac{\bar\nu_0(x)}{\left(P_{n-1}(d\nu_0;x)\frac{\varphi(d\nu_0;x)}{2}-Q_{n-1}(d\nu_0;x)\right)^2+\pi^2\bar\nu_0^2( x)P_{n-1}^2(d\nu_0;x)}
\end{equation}
$n=1,2,3,\ldots\,.$ Hence by observing definition (\ref{nu sequence defn}), we find Eq. (\ref{nu lemma eq}) for $n=1,2,3,\ldots\,.$\qed
\end{proof}
\begin{theorem}
A sequence of beta normalised measures $d\nu_0(x),$ $d\nu_1(x),$ $d\nu_2(x),$ $\ldots$, satisfy
\begin{eqnarray}
\alpha_{m+n}(d\nu_0)&=&\alpha_n(d\nu_m)\quad n,m=0,1,2,\ldots\,,\label{betanorm alpah n and m eq}\\
\beta_{m+n}(d\nu_0)&=&\beta_n(d\nu_m)\quad n,m=0,1,2,\ldots\,.\label{betanorm beta n and n eq}
\end{eqnarray}
\end{theorem}
\begin{proof}
From lemma (\ref{beta measures proportional to norm measures lemma}) we see that $d\mu_n$ and $d\nu_n$ are related by a constant, hence using lemmas (\ref{the shift property of alpha beta}) and (\ref{the norm sequence shif theorem v2}) we find
\begin{eqnarray}
\alpha_{m+n}(d\nu_0)&=&\alpha_n(d\nu_m)\quad n,m=0,1,2,\ldots\,,\\
\beta_{m+n}(d\nu_0)&=&\beta_n(d\nu_m)\quad n=1,2,3,\ldots\,,\, m=0,1,2,\ldots\,.
\end{eqnarray}
For $\beta_0(d\nu_m)$ we find
\begin{equation}
\beta_0(d\nu_m)=\beta_0(\beta_m(d\nu_0)d\mu_m)=\beta_{m}(d\nu_0)C_0(d\mu_m)\quad m=0,1,2,\ldots\,,
\end{equation}
where we have used Eq. (\ref{nu lemma eq}) followed by Eq. (\ref{the beta 0 shift eq}) and then corollary (\ref{cor C0 beta0}). But $C_0(d\mu_m)=1\quad m=0,1,2,\ldots\,,$ by definition.\qed
\end{proof}
\begin{remark}
As we will see in section \textup{\ref{examples section}}, for a wide range of $d\nu_0$\textup{;} $\alpha_{m+n}(d\nu_0)$ ,  $\beta_{m+n}(d\nu_0)$ and $d\nu_m$ can be determined analytically. Hence Eq.
 \textup{(\ref{betanorm alpah n and m eq})} and Eq. \textup{(\ref{betanorm beta n and n eq})} can also be used to find analytical solutions to a wide range of integrals.
\end{remark}
\begin{definition}
We will call the infinite tridiagonal matrix of a measure $d\mu(x)$
\begin{equation}
\mathcal{J}(d\mu)=\begin{bmatrix}\alpha_0(d\mu) & \sqrt{\beta_1(d\mu)} & & & & 0 \\
\sqrt{\beta_1(d\mu)} & \alpha_1(d\mu) & \sqrt{\beta_2(d\mu)} &  &\\
 & \sqrt{\beta_2(d\mu)} & \alpha_2(d\mu) & \sqrt{\beta_3(d\mu)} &\\
 &  & \ddots\ & \ddots\ &\ddots\ \\
0 &  & & & & \\
\end{bmatrix},
\end{equation}
the \textup{Jacobi matrix}. See \textup{\cite{Gautschi}} for more details.
\end{definition}
\begin{definition}
We will call the matrix
\begin{equation}\label{the mth associated jacobi matrix}
\mathcal{J}_{n}(d\mu)=\begin{bmatrix}\alpha_n(d\mu) & \sqrt{\beta_{n+1}(d\mu)} & & & & 0 \\
\sqrt{\beta_{n+1}(d\mu)} & \alpha_{n+1}(d\mu) & \sqrt{\beta_{n+2}(d\mu)} &  & &\\
 & \sqrt{\beta_{n+2}(d\mu)} & \alpha_{n+2}(d\mu) & \sqrt{\beta_{n+3}(d\mu)} & &\\
 &  & \ddots\  & \ddots\  & \ddots\ &\\
0 &  & & & & \\
\end{bmatrix}
\end{equation}
$n=1,2,3,\ldots\,,$ the \textup{nth associated Jacobi matrix} of the Jacobi matrix $\mathcal{J}(d\mu)$.
\end{definition}
We are now ready to state our second main theorem:
\begin{theorem}\label{jacobi matrix theorem}For Jacobi matrices for which its corresponding measure defines a sequence of normalised secondary measures, there exist an infinite sequence of associated  Jacobi matrices corresponding to the sequence of normalised secondary measures. These matrices are formed by crossing-out the first row and column of the previous Jacobi matrix in the sequence:
\begin{equation}\label{jacobi matrix theorem eq}
\mathcal{J}_{n}(d\mu_0)=\mathcal{J}(d\mu_n)\quad n=1,2,3,...
\end{equation}
\end{theorem}
\begin{proof}
By equating the matrix elements in Eq. (\ref{jacobi matrix theorem eq}), we find Eq. (\ref{alpah n and 0 eq}) and (\ref{beta n and 0 eq}). Hence lemma (\ref{the norm sequence shif theorem v2}) implies Eq. (\ref{jacobi matrix theorem eq}).\qed
\end{proof}
\begin{corollary}\label{jacobi theorem corr}
Theorem \textup{(\ref{jacobi matrix theorem})} is also valid for any sequence of measures which are proportional to a sequence of normalised secondary measures
such as the beta normalised measures.\end{corollary}
\begin{proof}
Given that the Jacobi matrix does not contain the $\beta_0$ coefficient, the result follows easily from lemma (\ref{the shift property of alpha beta}).\qed
\end{proof}

\section{Chain mappings of open quantum systems and Markovian embeddings}\label{Chain mappings of open quantum systems and Markovian embeddings}
\subsection{Chain mappings}\label{Generalised bosonic mapping}
An open quantum system can be represented by a system plus bath (also known as environment) model introduced by Caldeira and Leggett \cite{cal and leg}. The Hilbert space of the Hamiltonian for this model $\mathcal{H}$, is the tensor product of the space of the quantum system wavefunctions and the Fock space for the bosonic bath. Formally, $\mathcal{H}=\mathcal{S} \otimes \Gamma (\mathfrak{h})$, where $\mathcal{S}$ is a separable Hilbert space describing the quantum system and $\Gamma (\mathfrak{h})$ is the bosonic Fock space\footnote{also known as symmetric Fock space} over the one particle space $\mathfrak{h}=L^2(\mathbb{R},dk)$, where $k$ is the boson momentum. This describes a field of scalar bosons. An element $\Psi$ of $\mathcal{H}$, is a sequence $\{ \Psi^{(n)} \}$ with $\Psi^{(n)}=\varphi\otimes \psi^{(n)}$, where $\psi=\{\psi^{(n)}\}\in \Gamma(\mathfrak{h})$, $\psi^{(n)}$ is on $\mathbb{R}^n$ and the domain of $\varphi\in \mathcal{S}$ is to be specified with the details of the quantum system (we will give some examples later). The elements of $\mathcal{H}$ satisfy $||\Psi||:=||\varphi||_\mathcal{S}||\psi||_{\Gamma (\mathfrak{h})}<\infty$, where
\begin{equation}
||\Psi||^2=||\varphi||^2_\mathcal{S}\left( |\psi^{(0)}|^2+\sum_{n=1}^\infty \int_{k_{min}}^{k_{max}}\ldots\int_{k_{min}}^{k_{max}} dk_1\ldots dk_n |\psi^{(n)}(k_1,\ldots,k_n)|^2 \right),
\end{equation}
where each $\psi^{(n)}$ is symmetric in $k_1,\ldots,k_n$ and $-\infty \leq k_{min}<k_{max} \leq \infty$.
The Hamiltonian is
\begin{equation}\label{the mother ham}
H=H_{\mathcal{S}}\otimes\mathbb{I}_{\Gamma(\mathfrak{h})}+\mathbb{I}_{\mathcal{S}}\otimes \int_{k_{min}}^{k_{max}}dkg(k)a_k^* a_k+A\otimes\int_{k_{min}}^{k_{max}}dk h(k) (a_k^*+a_k),
\end{equation}
$H_S$ is a bounded below self-adjoint operator on $\mathcal{S}$ and describes the system dynamics, $H_E:=\int_{k_{min}}^{k_{max}}dxg(k)a_k^* a_k$ is the Hamiltonian of the environment where $g\geq 0$. This is also known in the literature as $d\Gamma(g)$, \textit{the second quantisation} of $g$. $a_k^*,$ $a_k$ are creation and annihilation operators with cummutator $[a_k,a_{k'}^*]=\delta(k-k')$. They act on each $\psi^{(n)}$ by
\begin{eqnarray}
(a_k\psi)^{(n)}(k_1,\ldots,k_n)&=&(n+1)^{1/2}\psi^{(n+1)}(k,k_1,\ldots,k_n)\quad n=0,1,2,\ldots\,,\\
(a^*_k\psi)^{(n)}(k_1,\ldots,k_n)&=&n^{-1/2}\sum_{j=1}^n\delta(k-k_j)\psi^{(n-1)}(k_1,\ldots,\hat k_j,\ldots,k_n)\quad n=1,2,3,\ldots\,,\nonumber\\
(a_k^*\psi)^{(0)}&=&0,\\
\end{eqnarray}
where $\hat k_j$ indicates that $k_j$ is omitted. Hence we have that $(\mathbb{I}_\mathcal{S}\otimes H_E\Psi)^{(n)}=\sum_{j=1}^ng(k_j)\Psi^{(n)}$ on domain $\mathcal{D} (\mathbb{I}_\mathcal{S}\otimes H_E)$ of all $\Psi\in \mathcal{H}$ such that $\{ (\mathbb{I}_\mathcal{S}\otimes H_E\Psi)^{(n)} \}$ is again in $\mathcal{H}$. Let $\mathfrak{n}$ be the number of bosons operator defined by
\begin{equation}
(\mathfrak{n}\psi)^{(n)}=n\psi^{(n)}, \quad n=0,1,2,\ldots
\end{equation}
on the domain $\mathcal{D}(\mathfrak{n})$ of all $\psi$ in $\Gamma(\mathfrak{h})$ such that $\{ n\psi^{(n)} \}$ is again in $\Gamma(\mathfrak{h})$. $H_{int}:=A\otimes\int_{k_{min}}^{k_{max}}dk h(k)(a_k^*+a_k)$ describes the interaction between the quantum system and the bosonic environment. $h(k)\in L^2([k_{min},k_{max}])$ and $\int_{k_{min}}^{k_{max}}dk h(k)(a_k^*+a_k)=\int_{k_{min}}^{k_{max}}dk h(k)a_k^*+ \int_{k_{min}}^{k_{max}}dk h(k)a_k$, where
\begin{eqnarray}
\left(\int_{k_{min}}^{k_{max}}dk h(k)a_k\psi\right)^{(n)}(k_1,\ldots,k_n)&=&(n+1)^{1/2}\int_{k_{min}}^{k_{max}}h(k)\psi^{(n+1)}(k,k_1,\ldots,k_n)dk\quad n=0,1,2,\ldots\,,\label{int a* eq}\nonumber\\
\left(\int_{k_{min}}^{k_{max}}dkh(k)a^*_k\psi\right)^{(n)}(k_1,\ldots,k_n)&=&n^{-1/2}\sum_{j=1}^nh(k_j)\psi^{(n-1)}(k_1,\ldots,\hat k_j,\ldots,k_n)\quad n=1,2,3,\ldots\,,\label{int a eq}\\
\left(\int_{k_{min}}^{k_{max}}dkh(k)a^*_k\psi\right)^{(0)}&=&0,\label{int ao eq}
\end{eqnarray}
is a well defined self-conjugate operator (also known in the literature as $\Phi_s(h\sqrt{2})$, the \textit{Segal field operator}) on $\mathcal{D}(\mathfrak{n}^{1/2})$. $A$ is any bounded below self-adjoint operator with domain in $\mathcal{S}$ such that $H$ is a well defined Hamiltonian with domain in $\mathcal{H}$. See Theorem (\ref{selfadjoint}) and remark (\ref{selfadjoint remark}) for more details. We call $H_0=H_{\mathcal{S}}\otimes\mathbb{I}_{\Gamma(\mathfrak{h})}+\mathbb{I}_{\mathcal{S}}\otimes H_E$ the \textit{free Hamiltonian} and define $\omega_{max}:=\textup{sup}\, g$, $\omega_{min}:=\text{inf}\, g$. $\omega_{min}$ is sometimes called as the mass of the bosons. The case of \textit{massless bosons}, i.e. when $\omega_{min}=0$, is of particular interest.
\begin{theorem}\label{selfadjoint}
If $A$ is a bounded operator \footnote{i.e. $||A\varphi||_\mathcal{S}\leq C||\varphi||_\mathcal{S} \forall \varphi\in \mathcal{D}(A),\,\,0\leq C<\infty$.} on $S$, $H$ is self-adjoint on $\mathcal{D}(H)=\mathcal{D}(H_0)$ if
\begin{equation}\label{selfadjon eq}
\int_{k_{min}}^{k_{max}}\frac{h^2(k)}{g(k)}dk<\infty
\end{equation}
\end{theorem}
\begin{proof} This is a well-known result from the literature. See e.g. \cite{Be}, section 7.
\end{proof}
\begin{remark}\label{selfadjoint remark}
In the case that $A$ is an unbounded operator on $\mathcal{S}$, there also exist conditions under which $H$ can be shown to be self-adjoint on some appropriate domain. See Theorem 2.2 of \cite{Jan}.
\end{remark}
Two examples of Hamiltonians satisfying Lemma (\ref{selfadjoint}) are the Spin-Boson model, where $H_\mathcal{S}=\mathcal{\alpha}\sigma_z$, $A=\sigma_x,$ where $\mathcal{\alpha}$ is a positive constant and $\sigma_x,$ $\sigma_z$ are the Pauli matrices. $\varphi=(\gamma_1,\gamma_2)^T$ with $\gamma_1,\gamma_2\in \mathbb{C}$. Another example is when the quantum system is comprised of $N$ spinless nucleons of mass $M>0$. In this case $H_\mathcal{S}=-\sum_{j=1}^N\nabla^2_j/(2M)$ where $\nabla^2_j$ is the Laplacian in the variable $x_j$ with usual domain which makes it a self-adjoint operator and $\varphi=\varphi(x_1,\ldots,x_N)$. $A$, for example could be the operator of multiplication by a positive function $f\in L^2( \mathbb{R}^N)$.\\
\textbf{Additional assumptions:} In order to prove the results of this paper, we will need to make some assumptions in addition to those which make the Hamiltonian (\ref{the mother ham}) well-defined:
\begin{itemize}
\item [A1.]  $g(x)$ is invertable and differentiable on the interval $[k_{min},k_{max}]$ satisfying $dg(x)/dx\geq 0$ or $dg(x)/dx\leq  0$. \footnote{Note however that if this is not the case, there should be ways to get around this difficulty.}
\item [A2.] $J(x)$ has finite moments  on the interval $[\omega_{min}, \omega_{max}]$.\footnote{Note that if this is not the case because $[\omega_{min}, \omega_{max}]$ is an unbounded interval, one can define a cut-off hamiltonian such that the interval is bounded and the theorems developed here will then apply.}
\end{itemize}
\begin{corollary}
If $A$ is a bounded operator on $\mathcal{S}$, and assumptions $A1.$ and $A2.$ are satisfied, then $H$ is self-adjoint on $\mathcal{D}(H)=\mathcal{D}(H_0)$. 
\end{corollary}
\begin{proof}
We can use $A1$ to make the change of variable $k=g^{-1}(x)$ in Eq. (\ref{selfadjon eq}) and hence write the integrand in terms on J. Now $A2.$ implies that Theorem (\ref{selfadjoint}) is satisfied.
\end{proof}
Physically $g(x)$ represents the dispersion relation of the environment and $h(x)$ determines the system-environment coupling strength.  Together they determine the spectral density in the following way
\begin{definition}\label{SD eq defn} We call the function $J(x)$ the spectral density,
\begin{equation}\label{SD_eq}
J(\omega)=\pi h^2\left(g^{-1}(\omega)\right)\left|\frac{dg^{-1}(\omega)}{d\omega}\right|,
\end{equation}
where $g^{-1}(g(x))=g(g^{-1}(x))=x$, \textup{Dom}$[J]\in [\omega_{min},\omega_{max}]$, and $|.|$ denotes the absolute value. Recall that $\omega_{max}$ can be finite or infinite such that the interval $[\omega_{min},\omega_{max}]$ can be bounded or unbounded.\\
Analogously to the measures \textup{(definition (\ref{meausre defn}))}, we call the spectral density \textup{gapless} if $J$ has support on all of the interval $(\omega_{min},\omega_{max})$ and \textup{gapped} if the support of $J$ is a disjoint set of intervals.
\end{definition}
\begin{remark}
As we will see later in the proof to theorem (\ref{generalised mapping theorem}) (more specifically in Eq. \ref{the mesure again}), we use $J$ in the definition of a measure with weight function $M^q(x)$ in such a way that the measure will be gapped iff $J$ is gapped. Hence if we want to form a gapless measure, in accordance with definition (\ref{meausre defn}), we need to use a gapless spectral density. For cases where this does not hold, i.e. $J$ is gapped, we can define a set of spectral densities corresponding to gapless measures by redefining their domain such that they contain only non positive values at the boundaries of their domain. It is not necessary to do this, but has some advantages as some additional theorems will then apply. Also see remark (\ref{gaps}).
\end{remark}

\begin{figure}
\centering 
\includegraphics[scale=0.4]{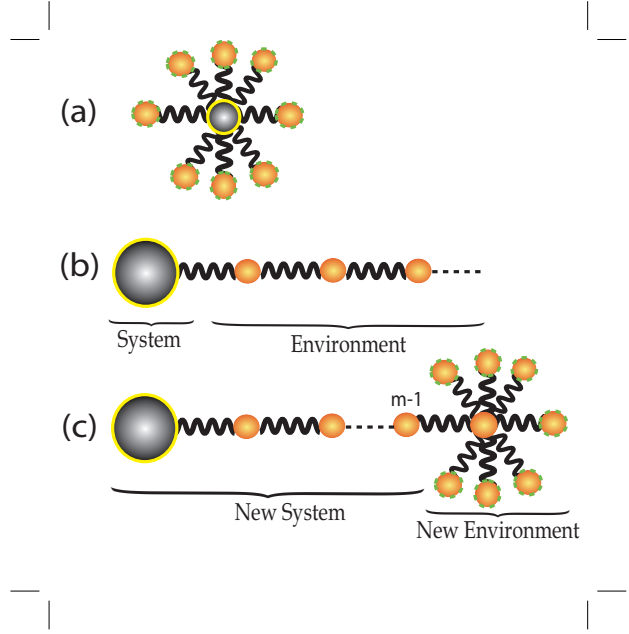} 
\caption{\label{fig:myFig} \textbf{(a)} The initial system-environment Hamiltonian before the mapping has been performed: the system (depicted by the gray ball in the centre) couples directly to the degrees of freedom of the enviroment (orange balls).
\textbf{(b)} The system-environment Hamilatonian after the chain mapping has been performed: the environment has been mapped onto a semi-infinite chain of nearest neighbour interactions where the system now only couples to the first element in the chain.
\textbf{(c)} The system-environment Hamiltonian after the $m$th environmental degree of freedom has been embedded: the new system, interacts with the new  enviroment via the $m$th residual spectral density $J_m(x)$.}
\end{figure}

\begin{lemma}\label{generalised orthogonality lem}
Given a measure $d\mu(x)$ with $I=[G_q(g(k_{min})), G_q(g(k_{max}))]$ if $g$ is non-decreasing and $I=[G_q(g(k_{max})), G_q(g(k_{min}))]$ if $g$ is non-increasing and its corresponding monic orthogonal polynomials $\{\pi_n( {d\mu};x)\}_{n=0}^{n=\infty}$, the following holds. One can construct the set of functions $\{\underline\pi_n(\underline {d\mu};x)\}_{n=0}^{n=\infty}$, 
\begin{equation}
\underline \pi_n (d \underline \mu;x):= \pi_n\big(d\mu; G_q(g(x))\big), \quad n=0,1,2,\ldots\,,
\end{equation}
where
\begin{equation}
G_q(x):=\frac{-q(1+q^2)+2\sqrt{q^4+4(1-q^2)x^2}}{ 4(1-q^2)} \quad x\geq0,\,q \in [0,1]
\end{equation}
which satisfy the 3-term recurrence relation
\begin{eqnarray}\label{genralised tree recurrence rel}
\underline\pi_{n+1}(\underline {d\mu};x)&=\left(G_q(g(x))-\alpha_n(d\mu)\right)\underline\pi_n(\underline {d\mu};x)-\beta_n(d\mu)\underline\pi_{n-1}(\underline {d\mu};x),& \label{3-term recurrece eq 3}\\
&\underline\pi_{-1}(\underline {d\mu};x):=0 \quad n=0,1,2,\ldots\,
\end{eqnarray}
and are orthogonal with respect to the measure $d \underline \mu(x)$
\begin{equation}\label{new ortho real}
\int_{k_{min}}^{k_{max}} \underline\pi_n(\underline {d\mu};x) \underline\pi_m(\underline {d\mu};x)\underline{\bar \mu}(x)dx=\langle \pi_n(d\mu),\pi_n(d\mu)\rangle_{\bar\mu}\delta_{nm} \quad n,m=0,1,2,\ldots\,,
\end{equation}
where $\underline {d\mu}(x)=:\underline{\bar \mu}(x)dx$, with
\begin{equation}\label{new measure inroems of alod}
\underline{\bar \mu}(x)=\bar \mu\left(G_q(g(x))\right)\left|\frac{dG_q(g(x))}{dx}\right|,
\end{equation}
with $I=[k_{min}, k_{max}]$. Furthermore, the converse is true: given the set of functions $\{\underline\pi_n(\underline{d\mu};x)\}_{n=0}^\infty$ defined in terms of a set of monic polynomials $\{\pi_n(x)\}_{n=0}^\infty$ which satisfy Eqs. (\ref{genralised tree recurrence rel}) and (\ref{new ortho real}) for a weight function $\underline{\bar\mu}$ defined in terms of a weight function $\bar\mu$ with $I=[G_q(g(k_{min})), G_q(g(k_{max}))]$ if $g$ is non-decreasing and $I=[G_q(g(k_{max})), G_q(g(k_{min}))]$ if $g$ is non-increasing, then the set $\{\pi_n(x)\}_{n=0}^\infty$ are the corresponding monic orthogonal polynomials for the weight function $\bar\mu$.
\end{lemma}
\begin{proof}
Eqs. (\ref{new ortho real}) and (\ref{new measure inroems of alod}) follow from perfoming a change of variable in \\$\langle \pi_n(d\mu),\pi_m(d\mu)\rangle_{\bar\mu}=\langle \pi_n(d\mu),\pi_n(d\mu)\rangle_{\bar\mu}\delta_{nm}$ and noting that $dG_q(g(x))/dx\geq 0$ if $g$ is non-decreasing and $dG_q(g(x))/dx\leq 0$ if $g$ is non-increasing. Eq. (\ref{genralised tree recurrence rel}) follows from a change of variable in Eq. (\ref{3-term recurrece eq}). The fact that the converse is true follows from the fact that $G_q$ and $g$ have well-defined inverse functions.
\end{proof}
\begin{theorem}\label{generalised mapping theorem}
There is a class of quantum systems linearly coupled with a reservoir with
spectral density $J(x)$ which are equivalent to semi-infinite chains with only nearest-neighbors interactions, where the system only couples to the first site in the chain. More specifically, starting from an initial Hamiltonian,
\begin{equation}\label{the 1st eq of genre theorem}
H= H_S\otimes\mathbb{I}_{\Gamma(\mathfrak{h})}+\mathbb{I}_\mathcal{S}\otimes\int_{k_{min}}^{k_{max}}dkg(k)a_k^* a_k+A\otimes\int_{k_{min}}^{k_{max}}dk h(k)(a_k^*+a_k),
\end{equation}
satisfying the assuptions $A1,A2$, there exists, for every $q\in[0,1],$ a countably infinite set of new creation $b_n^*(q)$ and annihilation $b_n(q)$ operators
\begin{eqnarray}
b^*_n(q) &=&\int_{k_{min}}^{k_{max}}dxU_n(d\lambda^q;x)\left[\cosh r_q\big(g(x)\big)a^*_x-\sinh r_q\big(g(x)\big) a_x\right]\quad n=0,1,2,\ldots\,,\\
b_n(q) &=&\int_{k_{min}}^{k_{max}}dxU_n(d\lambda^q;x)\left[\cosh r_q\big(g(x)\big)a_x-\sinh r_q\big(g(x)\big) a_x^*\right],
\end{eqnarray}
well-defined on $\mathcal{D}(\mathfrak{n}^{1/2})$
which satisfy the commutation relations
\begin{equation}
[b_n(q), b^*_m(q)] = \delta_{nm}\quad n,m=0,1,2,\ldots\,,
\end{equation}  with transformed Hamiltonian
\begin{eqnarray}\label{transform general map eq}
:\!H\!:&=& H_S\otimes\mathbb{I}_{\Gamma(\mathfrak{h})}+\mathbb{I}_\mathcal{S}\otimes H_{E,q}+H_{int,q}\,,\\
H_{E,q}&=&\sum_{n=0}^\infty \Big\{E_{1n}(q)(b_n^*(q) b_n^*(q)+b_n(q) b_n(q))+E_{2n}(q)b^*_n(q)b_n(q) \\
&&+E_{3n}(q)(b_n^*(q) b_{n+1}^*(q)+b_n(q) b_{n+1}(q))\\
&&+E_{4n}(q)(b_n^*(q) b_{n+1}(q)+b_n(q) b^*_{n+1}(q)) \Big\},\\
H_{int,q}&=&E_{5}(q) A\otimes(b_0^*(q)+b_0(q)).
\end{eqnarray}
where $E_{1n}(q)$, $E_{2n}(q)$, $E_{3n}(q)$, $E_{4n}(q)$, $E_{5}(q)$ $\in\mathbb{R}$ and $q\,\in[0,1]$ is a free parameter of the mapping which determines the particular version. $:\!H\!:$ is to indicate that the hamiltonian $H$ has been renormalized\footnote{Specifically, a constant factor $C \mathbb{I}_\mathcal{S}\otimes\mathbb{I}_\Gamma{(\mathfrak{h})},$ $C\in\mathbb{R}$ has been neglected.}. The constants $E_{1n}(q)$, $E_{2n}(q)$, $E_{3n}(q)$, $E_{4n}(q)$, $E_{5}(q)$ are determined when a particular spectral density $J(x)$ and value of $q$ are specified. Expressions for the functions $U (d\lambda^q; x)$ and $r_q(x)$ are derived in the proof. \textup{For a pictorial representation of this theorem, see \textbf{(a)} and \textbf{(b)} of figure (\ref{fig:myFig}).}
\end{theorem}

\begin{proof}
The proof is by construction. Let us start by defining a local transformation of the creation and annihilation operators $a_x^*$ and $a_x$ into another set, $c_x$ and $c^*_x$, which preserves the commutation relations. We can do this via a so called \textit{Bogoliubov transformation}\cite{bogo}
\begin{eqnarray}
a_x&=&\cosh\big(r_q(g(x))\big)c_x+\sinh\big(r_q(g(x))\big)c_x^* \label{bogolu1},\\
a_x^*&=&\cosh\big(r_q(g(x))\big)c_x^*+\sinh\big(r_q(g(x))\big)c_x,\label{bogolu2}
\end{eqnarray}
where $r_q(x)\in  \mathbb{R}$ and $[a_x,a^*_y]=[c_x,c^*_y]=\delta(x-y)$. Notice how we have parametrised the argument of the $\cosh$ and $\sinh$ functions in terms of $r_q(g(x))$. As we will see later, there are a familiy of functions $r_q(x)$ for which it is usefull to perform this tranformation. The $q$ is to denote which particular function is being used and will become clear soon. \footnote{We can replace the $cosh$ and $sinh$ functions with $cos$ and $sin$ functions for fermions. By doing so, we could work out a different version of this theorem which would be valid for the case where the initial bosonic creation-anhilation operators $\{a_x, a_x^*\}$ $x\in [k_{min}, k_{max}]$ where fermionic instead}
If we now parameterise $r_q(x)$ by introducing a new function $\xi_q(x)$ through $r_q(x)=\ln \xi_q(x)$ where $\xi_q(x)\in  [0,+\infty)$ using Eq. (\ref{bogolu1}) and (\ref{bogolu2}) we find after renormalising
\begin{eqnarray}\label{interaction temrms +bath}
 :\!\int_{k_{min}}^{k_{max}}dxg(x) a_x^{*} a_x\!:\,&=&\int_{k_{min}}^{k_{max}}dx\frac{g(x)}{4\xi_q^2(g(x))}\Big( (\xi_q^4(g(x))-1)(c_x^* c_x^*+c_x c_x)\nonumber \\&&+2(\xi_q^4(g(x))+1)c_x^* c_x \Big).
\end{eqnarray}
System-environment term simplifies to
\begin{equation}\label{general counplig term}
A\otimes\int_{k_{min}}^{k_{max}}dxh(x)(a_x+a_x^{*})=A\otimes\int_{k_{min}}^{k_{max}}dx\,h(x)\xi_q(g(x))(c_x^*+c_x).
\end{equation}
For appropriate choice of the function $\xi_q(x)$, we can define a measure 
\begin{eqnarray}
d\underline \lambda^q(x)&=&\underline M^q(x)dx,\label{generalised mesure}\\
\underline M^q(x)&=&h^2(x)\xi_q^2(g(x)).\label{generalised mesure234}
\end{eqnarray}
With the set of monic polynomials $\{\pi_n(d\lambda^q ;x)\}_{n=0}^\infty$ with measure $d\lambda^q(x)$, we are able to construct the set $\{\underline \pi_n(d\underline \lambda^q ;x)\}_{n=0}^\infty$ of orthogonal functions with respect to measure $d\underline \lambda^q(x)$ defined in Lemma \ref{generalised orthogonality lem}. We can use these to define the set of functions $\{U_n(d\lambda^q; x)\}_{n=0}^\infty$ through the relation
\begin{equation}
U_n(d\lambda^q;x)=\frac{\underline\pi_n(d\underline\lambda^q;x)\sqrt{\underline M^q(x)}}{\sqrt{\langle \underline\pi_n(d\underline\lambda^q),\underline\pi_n(d\underline\lambda^q)\rangle_{\underline M^q}}}=\frac{\underline\pi_n(d\underline\lambda^q;x)h(x)\xi_q(g(x))}{\sqrt{\langle \underline\pi_n(d\underline\lambda^q),\underline\pi_n(d\underline\lambda^q)\rangle_{\underline M^q}}}.
\end{equation}
We can now define the set of creation and annihilation operators of the chain (we will specify their domain and show that they are well-defined later in the proof)
\begin{eqnarray}
b_n^*(q)&=&\int_{k_{min}}^{k_{max}}dxU_n(d\lambda^q;x)c_x^*(q),\label{b n * Un}\\
b_n(q)&=&\int_{k_{min}}^{k_{max}}dxU_n(d\lambda^q;x)c_x(q).\label{b n  Un}
\end{eqnarray}
Substituting the inverse relations
\begin{eqnarray}
c_x^*(q)&=&\sum_{n=0}^\infty U_n(d\lambda^q;x)b_n^*(q),\label{c_x inver rel}\\
c_x(q)&=&\sum_{n=0}^\infty U_n(d\lambda^q;x)b_n(q)\label{c^dag_x inver rel},
\end{eqnarray}
into Eq. (\ref{general counplig term}) we obtain
\begin{equation}\label{interaction sys env generalised mapping eq}
A\otimes\int_{k_{min}}^{k_{max}}dxh(x)(a_x+a_x^{*})=\sqrt{\beta_0(d\lambda^q)}\, A\otimes(b_0^*(q)+b_0(q)).
\end{equation}
So we note that for all functions $\xi_q(x)$ that result in a valid measure Eq. (\ref{generalised mesure}), one can achieve a coupling between system and reservoir which only interacts with the first element in the chain.

Now we will examine carefully what type of chain can be generated via this transformation.

Using the orthogonality conditions of the orthogonal polynomials and Eq. (\ref{c_x inver rel}) and (\ref{c^dag_x inver rel}), we can transform terms of the form $\int_{k_{min}}^{k_{max}}c_{ax}c_{bx}dx$ into terms of the form $\sum_{n=0}^\infty W_nb_{an}b_{bn}$ where the $W_n$'s are constants and the sub indices $a,b=0$ denote that the operator is an annihilation operator and $a,b=1$ denote that they are creation operators. Also, we can transform terms of the form $\int_{k_{min}}^{k_{max}}G_q(g(x))c_{ax}c_{bx}dx$ into $\sum_{n=0}^\infty W_nb_{an}b_{bn}+W_1b_{a(n+1)}b_{bn}+W_{2n}b_{an}b_{b(n+1)}$ by using the three term recurrence relations Eq. (\ref{genralised tree recurrence rel}) to eliminate the $G_q(g(x))$. We can also map terms of the form $\int_{k_{min}}^{k_{max}}\left(G_q(g(x))\right)^kc_{ax}c_{bx}dx\quad k=2,3,4...$ using the three term recurrence relations in Eq. (\ref{genralised tree recurrence rel}) $k$ times, but this would result in every chain site coupling to its $k$th nearest neighbours. Given this resoning, we propose a trial solution to Eq.  (\ref{interaction temrms +bath}) which will reduce it to a chain of nearest neighbour interactions. This is:
\begin{eqnarray}
\frac{g(x)}{4\xi_q^2(g(x))}(\xi_q^4(g(x))-1)&=&c_1+g_1G_q(g(x)),\label{fist tria 1}\\
\frac{g(x)}{2\xi_q^2(g(x))}(\xi_q^4(g(x))+1)&=&c_2+g_2G_q(g(x)).\label{fist tria 2}
\end{eqnarray}
Not all values of the real constants $c_1, c_2, g_1, g_2$ will result in a valid trial solution. We will parametrise a valid set of these constants in terms of $q$ and hence there will be a different valid solution for the functions $\xi_q$ and $G_q$ for different values of $q$. This is where the $q$ dependency enters in the proof. A valid and usefull parametrisation is $2g_1=q, g_2=1, 8c_1=-q^2, 4c_2=q$, $q \in [0,1]$. Solving Eqs. (\ref{fist tria 1}) and (\ref{fist tria 2}) for $G_q$ and $\xi_q$ for this parametrisation of the constants gives us
\begin{eqnarray}
G_q(x)&=&\frac{-q(1+q^2)+2\sqrt{q^4+4(1-q^2)x^2}}{ 4(1-q^2)}\quad x\geq0, \label{xi  g, g}\\
\xi_q(x)&=&\left[\frac{q(1-q)+(1+q)4G_q(x)}{q(1+q)+(1-q)4G_q(x)}\right]^{1/4}\quad x\geq0\label{xi  g, xi}.
\end{eqnarray}
Now that we have explicit expresions for $G_q$ and $\xi_q$, we will derive an expression for the measure $d\lambda^q(x)=M^q(x)dx$ from $d\underline \lambda^q=\underline M^q(x)dx$. From Eqs (\ref{generalised mesure234}), (\ref{generalised mesure}) and (\ref{new measure inroems of alod}), we have
\begin{equation}
h^2(x)\xi^2(g(x))=\underline M^q(x)=M^q\big(G_q\left(g(x)\right)\big)\Big|\frac{dG_q\left(g(x)\right)}{dx}\Big|.
\end{equation}
Solving this for $M^q$ using Eqs (\ref{xi  g, g}) ,(\ref{xi  g, xi}) and definition (\ref{SD eq defn}), we find
\begin{equation}\label{the mesure again}
M^q(x)=\frac{J\big(G_q^{-1}(x)\big)}{\pi}\frac{(1+q^2)q+4(1-q^2)x}{(1+q)q+4(1-q)x},
\end{equation}
where $G_q(G^{-1}_q(x))=G^{-1}_q(G_q(x))=x$ and is given by
\begin{equation}
G^{-1}_q(x)=\frac{1}{4}\sqrt{[q(1-q)+4(1+q)x][q(1+q)+4(1-q)x]}\quad x\geq\frac{-q(1-q)}{4(1+q)},
\end{equation}
and we recall from Lemma \ref{generalised orthogonality lem}, that the interval $I$ for $d\lambda^q=M^q(x)$ is $I=[G_q(g(k_{min})), G_q(g(k_{max}))]$ if $g$ is a non-decreasing function and $I=[G_q(g(k_{max})), G_q(g(k_{min}))]$ if $g$ is a non-increasing function.
Noting that $M^q$ has finite moments on $I$ if $J$ does on $[\omega_{min},\omega_{max}]$ and taking into account assumption $A2.$ and Eqs. (\ref{b n * Un}), (\ref{b n  Un}), (\ref{xi  g, g}), (\ref{xi  g, xi}), we find that $f_{1,q,n}(x):=U_n(d\lambda^q;x)\cosh r_q\big(g(x)\big)$, $f_{2,q,n}(x):=U_n(d\lambda^q;x)\sinh r_q\big(g(x)\big)$ $\in L^2([k_{min},k_{max}])$. $b_n^*(q)\Psi^{(n)}$ and $b_n(q)\Psi^{(n)}$ are defined as in Eqs (\ref{int a* eq}) to (\ref{int ao eq}) when exchanging $h$ for $f_{1,q,n}$ and $f_{2,q,n}$ accordingly, and hence $b_n^*(q)$ and $b_n(q)$ are well-defined operators on $\mathcal{D}(\mathfrak{n}^{1/2})$. Furthermore we can verify that they satisfy
\begin{equation}
\left( \Psi_1, b_n(q)\Psi_2\right)=\left(b_n^*(q) \Psi_1, \Psi_2\right),
\end{equation}
for all $\Psi_1,\Psi_2\in\mathcal{D}(\mathfrak{n}^{1/2})$ hence confirming that $b_n^*(q)$ is the adjoint of $b_n(q)$. The commutation relation $[b_n(q),b_m^*(q)]=\delta_{n,m}$ follows from the orthogonality conditions for the monic polinomials and the commutator relation $[a_k,a_{k'}^*]=\delta(k-k')$.
Now let us perform the trasnformation. After substituting Eqs. (\ref{c_x inver rel}) and (\ref{c^dag_x inver rel}) into Eq. (\ref{interaction temrms +bath}) and using the orthogonality conditions and 3-term recurrence relations of Lemma \ref{generalised orthogonality lem}, we find
\begin{eqnarray}\label{generalised enviroment mapping eqarray}
&&:\!\int_{k_{min}}^{k_{max}}dxg(x) a_x^{*} a_x\! :\\
&=&\sum_{n=0}^\infty \Big\{\left( \frac{q}{2}\alpha_n(q)-\frac{q^2}{8} \right)(b_n^*(q) b_n^*(q)+b_n(q) b_n(q))+ \left( \alpha_n(q)+\frac{q}{4} \right)b^*_n(q)b_n(q) \quad\quad \\
&+& \sqrt{\beta_{n+1}(q)}\,\Big(q(b_n^*(q) b_{n+1}^*(q)+b_n(q) b_{n+1}(q))+ (b_n^*(q) b_{n+1}(q)+b_n(q) b^*_{n+1}(q))\Big)  \Big\},\nonumber
\end{eqnarray} where $\alpha(q):=\alpha(d\lambda^q)$, $\beta(q):=\beta(d\lambda^q)$.\qed
\end{proof}

\begin{remark}
Note that the chain mapping, theorem (\ref{generalised mapping theorem}), the spectral density can represent continuous modes and/or discrete modes. The discrete modes are represented by dirac delta distributions and the continuous modes by continuous functions. In the case of $N$ discrete modes only, the chain mapping will map the system onto a chain with $N$ sites. Mathematically, the reason for this is because the inner product of the measure (\ref{def inner product}) in this case can only be defined for functions living in the space spanned by a set of $N$ orthogonal polynomials of finite degree (See \textit{discrete measure}, page 4 following Theorem 1.8 of \cite{Gautschi}). Physically, this is because the number of degrees of freedom corresponding to the Hamiltonian after and before the mapping have to be the same. Also see remark (\ref{gaps}) regarding other features of spectral densities.\\
In addition, it is also worth noting that if the system interacts linearly with more than one environment through different system operators $A$ for each environment, then one can easily generalise the results of the generalised mapping (theorem (\ref{generalised mapping theorem})) such that we can map the Hamiltonian onto a Hamiltonian where the system interacts with more than one chain (i.e. one chain for each environment).
\end{remark}
\begin{corollary}\label{particle mapping case}
The generalised mapping \textup{Eq. (\ref{transform general map eq})} reduces to
\begin{equation}\label{marticle mapping hamiltonian eq}
H = H_S\otimes\mathbb{I}_{\Gamma{(\mathfrak{h})}}+\sqrt{\beta_0(0)} A\otimes(b_0(0)+b_0^*(0))+\mathbb{I}_\mathcal{S}\otimes\sum_{n=0}^\infty \alpha_n(0)b_n^*(0) b_n(0) +\sqrt{\beta_{n+1}(0)}(b_{n+1}^*(0) b_n(0) +h.c.)
\end{equation}
when $q=0$. We also find
\begin{equation}
M^0(x)=\frac{J(x)}{\pi},\label{the measure for the particle mapping sd eq}
\end{equation}
with
\begin{equation}
I=[\omega_{min},\omega_{max}]
\end{equation}
and
\begin{eqnarray}
G_0(x)&=&x,\\
b_n^*(0)&=&\int_{k_{min}}^{k_{max}}\frac{dx\,\pi_n(d\lambda^0;g(x))}{\sqrt{\langle \pi_n(d\lambda^0),\pi_n(d\lambda^0)\rangle_{ M^0}}}\sqrt{\frac{J(g(x))}{\pi}\left|\frac{dg(x)}{dx}\right|}a_x^*,\quad n=0,1,2,\ldots\,.
\end{eqnarray}
If we set $g(x)=\kappa x,\; \omega_{min}=0$, for some $\kappa>0$, then this chain representation reduces to the result found in \textup{\cite{Javi,Alex}}. Otherwise, the chain coefficients are the same, but here the correct relation between the operators $\{a_x, a_x^*\}$ and $\{b_n(0)\}_{n=0}^\infty$ is given.
\end{corollary}
\begin{proof}
Follows from setting $q=0$ in theorem (\ref{generalised mapping theorem}) and simplifying the resultant expressions.\qed
\end{proof}
Given that the coupling of the chain elements is excitation number preserving, the elementary excitations of the chain can be viwed as particles hopping on a 1d lattice. We therefore make the following definition.
\begin{definition}\label{particle case def}
We shall refer to the transformation described in corollary \textup{(\ref{particle mapping case})} as the \textup{particle mapping}.
\end{definition}
\begin{corollary}\label{phonon mapping case}
When $q=1$, the generalised mapping \textup{Eq. (\ref{transform general map eq})} reduces to
\begin{equation}\label{phonon chain hamiltonian}
:\!H\!:\,= H_{S}\otimes\mathbb{I}_{\Gamma{(\mathfrak{h})}}+\sqrt{\beta_0(1)} A\otimes X_0+\mathbb{I}_\mathcal{S}\otimes\sum_{n=0}^{\infty} \left(\sqrt{\beta_{n+1}(1)}X_nX_{n+1}+ \frac{\alpha_n(1)}{2}X^2_n +\frac{1}{2}P_n^2\right).
\end{equation}
where $X_n$ and $P_n$ are position and momentum operators, $X_n:=(b_n^*(1)+b_n(1))$, $P_n:=i(b_n^*(1)-b_n(1))/2,$ on $\mathcal{D}(\mathfrak{n}^{1/2}),$  $ n=0,1,2,\ldots\,.$ We also find
\begin{equation}
M^1(x)=\frac{J(\sqrt{x})}{\pi},\label{M1 eq}
\end{equation}
with
\begin{equation}
I=[\omega_{min}^2,\omega_{max}^2]\label{min max int for phonon}
\end{equation}
and
\begin{eqnarray}
G_1(x)&=&x^2,\\
b_n^*(1)&=&\int_{k_{min}}^{k_{max}}\frac{dx\,\pi_n(d\lambda^1;g^2(x))}{\sqrt{\langle \pi_n(d\lambda^1),\pi_n(d\lambda^1)\rangle_{ M^1}}}\sqrt{\frac{J(g(x))}{\pi}\left|\frac{dg(x)}{dx}\right|}\Big(\frac{2g(x)+1}{2}a_x^* - \frac{2g(x)-1}{2}a_x\Big).\nonumber
\end{eqnarray}
\end{corollary}
\begin{proof}
Follows from setting $q=1$ in theorem (\ref{generalised mapping theorem}) and simplifying the resultant expressions.\qed
\end{proof}
Given that the coupling of the chain elements in Eq. (\ref{phonon chain hamiltonian}) resemble that of springs obeying hooks law, the elementary excitations are phonons such as in solid state physics. We therefore make the following definition.
\begin{definition}\label{phonon case def}
We shall refer to the transformation described in corollary \textup{(\ref{phonon mapping case})} as the \textup{phonon mapping}.
\end{definition}
\begin{remark}
In light of definitions \textup{(\ref{particle case def})} and \textup{(\ref{phonon case def})}, we note that Eq.  \textup{(\ref{transform general map eq})} interpolates between the two solutions.
\end{remark}

We will now re-write the generalised mapping in terms of Jacobi matrices, this is to illustrate the connection with Jacobi matrix theory and to write the Hamiltonian in a more compact form. 
\begin{corollary}\label{chain in jacobi for conc}
The generalised mapping in terms of Jacobi matrices is:
\begin{eqnarray}\label{the generalised mapping in terms of the jacobi mat eq}
H&=&H_{S}\otimes\mathbb{I}_{\Gamma{(\mathfrak{h})}}+\sqrt{\beta_0(q)}\, A\otimes(b_0^*+b_0)\\
&+&\mathbb{I}_\mathcal{S}\otimes\frac{q}{2}\left[ \vec b_0^T\left(\mathcal{J}(d\lambda^q)-\frac{q}{4}\mathbb{I}\right)\vec b_0 +h.c.\right]+\mathbb{I}_\mathcal{S}\otimes\vec b_0^* \left(\mathcal{J}(d\lambda^q) +\frac{q}{4}\mathbb{I}\right)\vec b_0,\nonumber
\end{eqnarray}
where
\begin{eqnarray}\label{the b n vector eq}
\vec b_n^*=\vec b_n^*(q)&:=&(b_n^*(q),b_{n+1}^*(q),b_{n+2}^*(q),b_{n+3}^*(q),\ldots)^T,\\
\vec b_n=\vec b_n(q)&:=&(b_n(q),b_{n+1}(q),b_{n+2}(q),b_{n+3}(q),\ldots)^T \quad n=0,1,2,\ldots\,.
\end{eqnarray}
\end{corollary}
Let us define (for every $q\in[0,1]$) the orthonormal Fock basis
\begin{equation}
F_q:=\Big{\{} \frac{{b^*_0}^{m_0}(q)|0\rangle}{\sqrt{m_0!}}\otimes\frac{{b^*_1}^{m_1}(q)|0\rangle}{\sqrt{m_1!}}\otimes \frac{{b^*_2}^{m_2}(q)|0\rangle}{\sqrt{m_2!}}\otimes \ldots\Big{\}}_{\{m_n \}_{n=0}^\infty=0}^\infty,
\end{equation}
where the creation/annihilation operators (which we recall satisfy $[b_n(q),b_m^*(q)]=\delta_{n,m},$ $[b_n(q),b_m(q)]=[b_n^*(q),b_m^*(q)]=0$) act on the kets $|n\rangle$ in the standard way: $b_m^*(q)|n\rangle=\sqrt{n+1}|n+1\rangle,$ $b_m(q)|n\rangle=\sqrt{n}|n-1\rangle.$ Let $\mathcal{K}_n$ be the space of all complex linear combinations of $\Big{\{} \frac{{b^*_n}^{m}(q)|0\rangle}{\sqrt{m!}}\Big{\}}_{m=0}^\infty$ such that $\langle \gamma|\gamma\rangle\ <\infty$ for all $|\gamma\rangle\in\mathcal{K}_n$.
\begin{definition}
We call \textup{$m$th embedded system} to the new system-environment interaction produced when the new system is composed of the initial system plus the first $m$ sites of the chain formed by the environment in the chain representation, this is to say, the quantum system described by the Hamiltonian $H_{S^q_m}$ on $\mathcal{S}\otimes_{n=0}^{m-1} \mathcal{K}_n,$
\begin{eqnarray}\label{H s q eq}
H_{S^q_m}&:=&H_{S}\otimes\mathbb{I}_{\Gamma{(\mathfrak{h})}}+\sqrt{\beta_{0}(q)}\, A\otimes(b_0^*+b_0)\\  &+&\mathbb{I}_\mathcal{S}\otimes\sum_{n=0}^{m-2} \sqrt{\beta_{n+1}(q)}\,\Big(q(b_n^* b_{n+1}^*+b_n b_{n+1})+ (b_n^* b_{n+1}+b_n b^*_{n+1})\Big)\quad\quad\\&+&\mathbb{I}_\mathcal{S}\otimes\sum_{n=0}^{m-1}\left( \frac{q}{2}\alpha_n(q)-\frac{q^2}{8} \right)(b_n^* b_n^*+b_n b_n)+ \left( \alpha_n(q)+\frac{q}{4} \right)b^*_nb_n.
\end{eqnarray}
The \textup{$m$th environment} is formed by the remaining environment terms in the Hamilitonian, in other words the quantum system described by the Hamiltonian $H_{E_m^q}$ on $\otimes_{n=m}^\infty \mathcal{K}_n,$
\begin{eqnarray}\label{jacobi rep of gen map}
H_{E_m^q}&=&\frac{q}{2}\left[ \vec b_m^T\left(\mathcal{J}(d\lambda^q_m)-\frac{q}{4}\mathbb{I}\right)\vec b_m +h.c.\right]+\vec b_m^* \left(\mathcal{J}(d\lambda^q_m) +\frac{q}{4}\mathbb{I}\right)\vec b_m \\ \quad m&=&1,2,3,\ldots\,.
\end{eqnarray}
\end{definition}
Hence we have
\begin{equation}\label{particle embedding eq}
H=H_{S^q_m}+\sqrt{\beta_{m}(q)}\,\Big(q(b_{m-1}^* b_{m}^*+b_{m-1} b_{m})+ (b_{m-1}^* b_{m}+b_{m-1} b^*_{m})\Big)+H_{E_m^q}.
\end{equation}
See figure \textbf{(c)} for a pictorial representation. As we will see in section \ref{convergence section}, the chain coefficients converge for a wide range of spectral densities, and hence all the specific features of an environment appear in the first sites of the chain. Consequently, these can be progressively (or directly, all in one go) absorbed into the system by making an embedding; to reduce the complexity of the effective environment.
\begin{definition}
We call the \textup{nth residual spectral density} $J_n(\omega)$ to the spectral density which describes the system-environment interaction of the nth embedding. We call the initial spectral density $J_0(\omega)$ such that $J_0(\omega)\equiv J(\omega)$. See figure (c).
\end{definition}

\subsection{Connection between the phonon mapping to previous work and the sequence of residual spectral densities}\label{Connection between the phonon mapping to previous work and the sequence of partial spectral densities}
\begin{theorem}\label{the partial sd generade from a measure in Bassano et all}
The sequence of residual gapless spectral densities
in \textup{\cite{Bassano}} are generated by
\begin{equation}\label{partial SD's phonon case}
J_n(\omega)=\frac{J_0(\omega)}{(P_{n-1}(d\lambda^1; \omega^2)\frac{\varphi(d\lambda^1;\omega^2)}{2}-Q_{n-1}(d\lambda^1; \omega^2))^2+ J_0^2(\omega)P_{n-1}^2(d\lambda^1; \omega^2)}
\end{equation}
$n=1,2,3,\ldots\,.$
\end{theorem}
\begin{proof}
As observed by Leggett \cite{legett prop},  when the spectral density of an open quantum system has support on a real interval (in contrast to having support on disjoint intervals), one can easily obtain it from it's  propagator $L_0(z)$. The authors of  \cite{Bassano} have developed this to find a continued fraction representation for the case of the mapping of the propagator presented in their paper as follows.
\begin{eqnarray}
L_0(z)&=&-z^2-w_0(z),\\\label{continued frac}
w_0(z)&=& \cfrac{D^2_0}{\Omega_1^2-z^2-\cfrac{D^2_1}{\Omega_2^2-z^2-\cfrac{D^2_2}{\Omega_3^2-z^2-\ldots}}} \,,
\end{eqnarray}
where
\begin{eqnarray}
D_n^2&=&\frac{2}{\pi}\int^\infty_0 d\omega J_n(\omega)\omega \quad n=0,1,2,\ldots\,, \label{D^2_n 1st} \\
\Omega^2_{n+1}&=&\frac{2}{\pi D^2_n}\int^\infty_0d\omega J_n(\omega)\omega^3 \quad n=0,1,2,\ldots\,,\label{omega^2_n 1st}
\end{eqnarray}
and $z\in \mathbb{C}-\,\textup{support}\left(J_0\right)$, where $\textup{support}\left(J_0\right)$ is an interval on the real line by definition \cite{Bassano}.
Alternatively, we note that we can write the continued fraction Eq. (\ref{continued frac}) as a recurrence relation
\begin{equation}\label{sequance eq}
w'_n(\sqrt{z})=\frac{D^2_n}{z-\Omega^2_{n+1}-w'_{n+1}(\sqrt{z})} \quad n=0,1,2,\ldots\,,
\end{equation}
where $w'_n(z):=-w_n(z)$. From \cite{Bassano}, we have an alternative expression for $w_n$ in terms of the residual spectral densities though the relation,
\begin{equation}\label{alsternative omnega n eq}
w_n(z)=\frac{2}{\pi}\int^\infty_0d\omega\frac{J_n(\omega)\omega}{\omega^2-z^2} \quad n=0,1,2,\ldots\,,
\end{equation}
which by a change of variables and taking into account the definition of $\omega'_n$ in Eq. (\ref{sequance eq}), can be written
in the form
\begin{equation}\label{int wprime}
w'_n(\sqrt{z})=\frac{1}{\pi}\int^\infty_0d\omega\frac{J_n(\sqrt{\omega})}{z-\omega} \quad n=0,1,2,\ldots\,.
\end{equation}
Furthermore, via a change of variables the integrals Eq. (\ref{D^2_n 1st}) and (\ref{omega^2_n 1st}) can be written as
\begin{eqnarray}\label{D_n int}
D_n^2&=&\frac{1}{\pi}\int^\infty_0 d\omega J_n(\sqrt{\omega}) \quad n=0,1,2,\ldots\,, \label{int Dn}\\
\Omega^2_{n+1}&=&\frac{1}{\pi D^2_n}\int^\infty_0d\omega J_n(\sqrt{\omega})\omega \quad n=0,1,2,\ldots\,.\label{int omegan}
\end{eqnarray}
We note that in \cite{Bassano} the support of the spectral densities corresponds with the domain of the spectral densities defined in definition (\ref{SD eq defn}) and hence integrals Eq. (\ref{int wprime}), (\ref{int Dn}) and  (\ref{int omegan}) are zero outside of the domain of $J_n(x),\quad n=0,1,2,\ldots$, therefore we can change the lower limit of $0$ and upper limit of $\infty$ of the integrals by $\omega^2_{min}$ and $\omega^2_{max}$ respectively. Now let us define the set of measures
\begin{eqnarray}
d\gamma_n(t)&=&dt\bar\gamma_ n(t) \quad n=0,1,2,\ldots\,,\\ \label{the SDs1}
\bar\gamma_n(t)&=&\frac{J_n(\sqrt{t})}{\pi D_n^2} \quad n=0,1,2,\ldots\,.\label{the SDs}
\end{eqnarray}
This definition of the measure has some important consequences:
\begin{itemize}
\item[1)] From Eq. (\ref{int omegan}) we note that $\Omega^2_{n+1}$ are the first moments of the measures $d\gamma_n(t)$,
\begin{equation}
C_1(d\gamma_n)=\Omega^2_{n+1}\quad n=0,1,2,\ldots\,.
\end{equation}
\item[2)]From Eq. (\ref{int wprime}) and definition (\ref{sitjes def}) we see that $w'_n(\sqrt{z})$ is proportional to the Stieltjes transformations of the measure $d\gamma_n(t)$
\begin{equation}
w'_n(\sqrt{z})=D^2_nS_{n}(z) \quad n=0,1,2,\ldots\,.
\end{equation}

\item[3)] From Eq. (\ref{int Dn}) we see that the zeroth moments of the measures  $d\gamma_n(t)$ are unity
\begin{equation}\label{normalised gamma sequence eq}
C_0(d\gamma_n)=1\quad n=0,1,2,\ldots\,.
\end{equation}
\end{itemize}

We are now able to re-write Eq. (\ref{sequance eq}) in the form
\begin{equation}\label{sitches rel}
S_{n+1}(z)D^2_{n+1}=z-C_1(d\gamma_n)-\frac{1}{S_{n}(z)}, \quad n=0,1,2,\ldots\,,
\end{equation}
where we have used the short hand $S_{m_n}(z)=:S_{n}(z)$.
By comparing this recursion relation with Eq. (\ref{eq sit rel}) and definition (\ref{def dary measure}), we deduce that $D^2_{n+1}d\gamma_{n+1}$ is the secondary measure associated with $d\gamma_n$, for $n=0,1,2,\ldots\,.$
We can also identify a sequence of normalised secondary measures. Noting that $C_0(d\gamma_0)=1$ from Eq. (\ref{normalised gamma sequence eq}), definition  (\ref{def 2nd dary measure}) tells us
that the sequence of secondary normalised measures starting from $d\gamma_0$ is
\begin{equation}
 d\gamma_0,\, D_1^2d\gamma_1/C_0(D_1^2d\gamma_1),\,  D_2^2d\gamma_2/C_0(D_2^2d\gamma_2),\ldots ,D_m^2d\gamma_m/C_0(D_m^2d\gamma_m),\ldots\,.
\end{equation}
 However,

\begin{eqnarray}
C_0(D_n^2d\gamma_n)&=&D_n^2C_0(d\gamma_n)=D_n^2\quad n=1,2,3\ldots\,,
\end{eqnarray}
so the sequence of normalised secondary measures is
 $d\gamma_{0}(t),$ $d\gamma_{1}(t),$ $d\gamma_{2}(t),$ $d\gamma_{3}(t),\ldots\,.$
Taking into account lemma (\ref{lemma continued frac}) and theorem (\ref{the dn beta theorem}) we see that
\begin{equation}
D_n^2=d_{n-1}=\beta_n(d\gamma_0)\quad n=1,2,3\ldots\,.
\end{equation}
Due to corollary (\ref{cor C0 beta0}) and Eq. (\ref{normalised gamma sequence eq}) we can also write $D^2_0$ in terms of $\beta_0$,
\begin{equation}\label{the d0 beta0 nu eq}
D^2_0= \beta_0(d\eta_0)
\end{equation}
where
\begin{equation}\label{the d d0 beta0 dnu eq}
d\eta_0:=D^2_0d\gamma_0.
\end{equation} We can now construct a sequence of beta normalised measures from $d\eta_0$, denoted by $d\eta_0,$ $d\eta_1,$ $d\eta_2,\ldots\,.$ From Eq. (\ref{the d0 beta0 nu eq}) and (\ref{the d d0 beta0 dnu eq}) we see that $d\eta_0$ satisfies Eq. (\ref{nu 0 eq}) for $d\gamma_0$, hence lemma (\ref{beta measures proportional to norm measures lemma}) tells us
\begin{equation}\label{nu gamma lemma eq}
\bar\eta_n(x)=\beta_n(d\eta_0)\bar\gamma_n(x)\quad n=0,1,2,\ldots\,.
\end{equation}
Taking into account that $\beta_n(d\eta_0)=\beta_n(d\gamma_0)\quad n=1,2,3,\ldots$ due to lemma (\ref{the shift property of alpha beta}) and Eq. (\ref{the d d0 beta0 dnu eq}), from Eq. (\ref{nu gamma lemma eq}) and (\ref{the SDs1}) we gather
\begin{eqnarray}
d\eta_n(t)&=&dt\,\bar \eta_n(t) \quad n=0,1,2,\ldots\,,\\ \label{the SDs2}
\bar \eta_n(t)&=&\frac{J_n(\sqrt{t})}{\pi} \quad n=0,1,2,\ldots\,,\label{the SDs2 }
\end{eqnarray}
hence we note that
\begin{equation}
d\eta_0=d\lambda^1.\label{the nu lambda equiv eq}
\end{equation} Substituting this into Eq. (\ref{the nu n eq}) gives us Eq. (\ref{partial SD's phonon case}).\qed
\end{proof}
In \cite{Bassano}, a similar mapping as in theorem (\ref{generalised mapping theorem}) is developed. Starting from a Hamiltonian of the form Eq. (\ref{the mother ham})\footnote{We say "of the form" because the authors do not define the hamiltonian rigorously.}, they show that it is equivalent to a Hamiltonian of the form
\begin{equation}\label{bassano eq}
H_S-D_0s\otimes X'_1+\mathbb{I}_\mathcal{S}\otimes\sum_{n=1}^{\infty}\left(-D_nX'_nX'_{n+1}+\frac{\Omega^2_n}{2}X^{\prime 2}_n+ \frac{1}{2}P_n^{\prime 2} \right),
\end{equation}
where $H_S= \frac{P^2}{2m}+V(s)+\Delta V(s)$ describes the quantum system dynamics and the self-adjoint operator $s$ couples the quantum system to the bosonic bath. The domain of the operators $\{X'_n,P'_n\}_{n=1}^\infty$ is not defined by the authors nor the precise relation to the operators in the initial Hamiltonian. See \cite{Bassano} for more details. Here $X_n$ and $P_n$ are position and momentum operators satisfying $[X_n,P_m]=i\delta_{n,m}$. The derivation of the coefficients $D_n,$ $\Omega_n$ involve repeated integration and apriori, seem unrelated to our method. However, given the apparent similarity between Eqs. (\ref{bassano eq}) and (\ref{phonon chain hamiltonian}), a more detailed analysis using some of the theorems developed in section 2 allows one to derive the following theorem which demonstrates that the the mapping of \cite{Bassano} is a special case of the phonon mapping, i.e. it is essentially the same as the phonon mapping for the case of a gappless spectral density.

\begin{theorem}\label{the partial sd Bassano et all equiv theorem}
Define $X'_{n}:=(-1)^nX_{n-1},$ $P'_{n}=(-1)^nP_{n-1},$ $s=A,$ $\mathcal{D}(X'_n)=\mathcal{D}(X_n),$ $\mathcal{D}(P'_n)=\mathcal{D}(P_n),$ $\mathcal{D}(s)=\mathcal{D}(A),$ $n=1,2,3,\ldots\,,$ and let the spectral density $J_0$ be gapless, then
\begin{eqnarray}
&&-D_0s\otimes X'_1+\mathbb{I}_\mathcal{S}\otimes\sum_{n=1}^{\infty}\left(-D_nX'_nX'_{n+1}+\frac{\Omega^2_n}{2}X^{\prime 2}_n+ \frac{1}{2}P_n^{\prime 2} \right)=\\
&&\sqrt{\beta_0(1)} A\otimes X_0+\mathbb{I}_\mathcal{S}\otimes\sum_{n=0}^{\infty} \left(\sqrt{\beta_{n+1}(1)}X_nX_{n+1}+ \frac{\alpha_n(1)}{2}X^2_n +\frac{1}{2}P_n^2\right).
\end{eqnarray}
\end{theorem}
\begin{proof}
From Eq. (\ref{int omegan}) and (\ref{the alpha def eq}) we see that we can write $\Omega_n^2$ as
\begin{equation}
\Omega_{n+1}^2=\alpha_0(d\eta_n)\quad n=0,1,2,\ldots\,.
\end{equation}
Hence,
\begin{equation}\label{the omega alpha eq}
\Omega_{n+1}^2=\alpha_n(d\eta_0)=\alpha_n(d\lambda^1)=\alpha_n(1)\quad n=0,1,2,\ldots\,,
\end{equation}
where we have used Eq. (\ref{betanorm alpah n and m eq}) followed by Eq. (\ref{the nu lambda equiv eq}). From Eq. (\ref{int Dn}), and (\ref{the beta 0 def eq}) we see that
\begin{equation}
D_n^2=\beta_0(d\eta_n)\quad n=0,1,2,\ldots\,.
\end{equation}
Hence  using Eq. (\ref{betanorm beta n and n eq}) followed by Eq. (\ref{the nu lambda equiv eq}), we find
\begin{equation}\label{Dn beta eq}
D_n^2=\beta_n(d\eta_0)=\beta_n(d\lambda^1)=\beta_n(1)\quad n=0,1,2,\ldots\,.
\end{equation}
Now using the definitions stated in the theorem, we conclude the proof.\qed
\end{proof}
\begin{remark}\label{gaps}
We note that the three term recurrence relations Eq. (\ref{3-term recurrece eq}) for gapped measures, still hold. Hence the generalised mapping (and therefore the phonon mapping) is still valid. However, the chain mapping presented in \cite{Bassano} is not valid under these conditions because the relation between the Stieltjes transformations of two consecutive measures, Eq. (\ref{eq sit rel}) is not valid anymore (see (\ref{break sidt}) in appendix for a proof) and hence one cannot calculate the chain coefficients from the sequence of residual spectral densities. In this sense, the phonon mapping presented here is a more general result. This is an important difference because if a spectral density is gapped, then the corresponding measure is also gapped. There are physical systems (such as photonic crystals and diatomic chains) which have these properties.
\end{remark}
\begin{corollary}\label{the phonon mapping corollary rof psd}
In the phonon mapping case, if the spectral density is gapless the sequence of residual spectral densities is given by
\begin{equation}
J_n(\omega)=\frac{J_0(\omega)}{(P_{n-1}(d\lambda^1;\omega^2)\frac{\varphi(d\lambda^1;\omega^2)}{2}-Q_{n-1}(d\lambda^1; \omega^2))^2+ J_0^2(\omega)P_{n-1}^2(d\lambda^1; \omega^2)}
\end{equation}
$n=1,2,3,\ldots\,.$
\end{corollary}
\begin{proof}
Follows directly from theorems (\ref{the partial sd generade from a measure in Bassano et all}) and (\ref{the partial sd Bassano et all equiv theorem}). Alternatively, we note that by applying the same line of reasoning of the proof for theorem (\ref{theorem for the sequence of sd in the particle m case}) to the phonon mapping Hamiltonian Eq. (\ref{phonon chain hamiltonian}), we can easily provide an alternative proof for this corollary. This has the advantage of being an independent derivation from the results of \cite{Bassano}, but has the downside of not illustrating the connections between this paper and their results.\qed
\end{proof}
\subsection{Sequence of residual spectral densities for the particle mapping case}\label{Sequence of partial spectral densities for the particle mapping case}
\begin{theorem}\label{theorem for the sequence of sd in the particle m case}
In the particle mapping case, if the spectral density is gapless the sequence of residual spectral densities is given by
\begin{equation}\label{partial SD's particle case}
J_n(\omega)=\frac{J_0(\omega)}{(P_{n-1}(d\lambda^0; \omega)\frac{\varphi(d\lambda^0;\omega)}{2}-Q_{n-1}(d\lambda^0; \omega))^2+ J_0^2(\omega)P_{n-1}^2(d\lambda^0; \omega)}
\end{equation}
$n=1,2,3,\ldots\,.$
\end{theorem}
\begin{proof}
The proof will start by defining a set of initially independent Hamiltonians. Each one of these will represent some system-environment interaction of a similar type discussed in this paper. When they are in this form, we can easily find an expression for the spectral density representing the system-environment interaction. We then put constrains on the Hamiltonians in such a way that they are no-longer independent from one another. Finally we show that these spectral densities correspond to the sequence of spectral densities for the particle mapping.\\
Let us start by defining a set of independent Hamiltonians on $\mathcal{H}$ labelled by $m$
\begin{equation}\label{the sequence of hamiltonians}
H_m=H_{S_m^0}+\mathbb{I}_\mathcal{S}\otimes\int_{k_{min}}^{k_{max}}g(x)a^*_{x,m}a_{x,m}dx+\mathbb{I}_\mathcal{S}\otimes\int_{k_{min}}^{k_{max}}h_m(x)(a_{x,m}^* b_{m-1}(0)+a_{x,m}b^*_{m-1}(0))dx
\end{equation}
$m=1,2,3\ldots,$ where $H_{S_m^0}$ is defined in (\ref{H s q eq}), the creation $b^*_{m}(0)$ annihilation $b_{m}(0)$ operators in theorem (\ref{generalised mapping theorem}) and the creation $a_{k,m}^*$, annihilation $a_{k,m}$ satisfy $[a_{k,m},a_{k',m}^*]=\delta(k-k')$ and will be defined later in the proof in terms of $b^*_{m}(0),b_{m}(0)$. For the particular choice of the functions $\{h_m\}_{m=1}^\infty$ that we will make in the proof, we will show that the $H_m$ are equal to the Hamiltonian (\ref{marticle mapping hamiltonian eq}), and hence self-adjoint on $\mathcal{H}$. We note that the $m$th Hamitonian $H_m$, has a spectral density $J^m(x)$ given by\footnote{We have used a superindex here rather than a subindex to denote the spectral densities so as not to confuse them with residual spectral densities, as at this stage we cannot identify them as such.} Eq. (\ref{SD_eq}):
\begin{equation}\label{SD_eq for mth particle}
J^m(\omega)=\pi h_m^2[g^{-1}(\omega)]\left|\frac{dg^{-1}(\omega)}{d\omega}\right|.
\end{equation}
Eq. (\ref{SD_eq for mth particle}) gives us $h^2_m(x)=J^m\left(g(x)\right)|dg(x)/dx|/\pi$.
We now define the set of measures
\begin{equation}\label{def of d theta for particle SD proof eq}
d\vartheta_m(x)=\frac{J^m\left(g(x)\right)}{\pi}dx\quad m=1,2,3,\ldots\,,
\end{equation}
and the new set of creation and annihilation operators
\begin{eqnarray}
b_{n+m}^{*}&=&\int_{k_{min}}^{k_{max}}dxU_n^m(x)a^*_{x,m},\\
b_{n+m}&=&\int_{k_{min}}^{k_{max}}dxU_n^m(x)a_{x,m}\quad n=0,1,2,\ldots,\,m=1,2,3,\ldots\,,
\end{eqnarray}
where $U_n^m(x)=h_m(x)P_n(d\vartheta_m;g(x))\quad n=0,1,2,\ldots\,,\, m=1,2,3,\ldots\,.$ Substituting the inverse relations
\begin{eqnarray}
a_{x,m}^{*}&=&\sum_{n=0}^\infty U_n^m(x)b_{n+m}^*,\\
a_{x,m}&=&\sum_{n=0}^\infty U_n^m(x)b_{n+m}\quad n=0,1,2,\ldots\,,\, m=1,2,3,\ldots\,,
\end{eqnarray}
into the RHS of Eq. (\ref{the sequence of hamiltonians}) and using the three term recurrence relations and orthogonality conditions in Lemma (\ref{generalised orthogonality lem}) for $q=0$, we find that we can write $H_m$ as
\begin{eqnarray}\label{sequence hams in chain rep}
H_m&=&H_{Sm}+\sqrt{\beta_0(d\vartheta_m)}(b_{m}^* b_{m-1}+b_{m}b_{m-1}^*)\\
&+&\sum_{n=0}^\infty\sqrt{\beta_{n+m+1}(d\vartheta_m)}(b^*_{n+m+1}b_{n+m}+h.c.)+ \alpha_n(d\vartheta_m)b_{n+m}^* b_{n+m}
\end{eqnarray}
$m=1,2,3,\ldots\,.$ We note that at this stage, the set of spectral densities $\{J^m(x)\}_{m=1}^\infty$ are independent and undefined. This freedom allows us to let the set of measures $\{d\vartheta_m \}_{m=1}^\infty$ be a sequence of beta normalised measures generated from the measure $d\lambda^0(x)=M^0(x)dx$   which from Eq.  (\ref{the measure for the particle mapping sd eq}) we see is given in terms of another spectral density $J(x)$. Now the spectral densities  $\{J^m(x)\}_{m=1}^\infty$ are fully determined by $J(x)$ through the definition of a beta normalised sequence of measures, definition (\ref{nu sequence defn}). Hence using Eq. (\ref{the nu n eq}) and (\ref{def of d theta for particle SD proof eq}) we find
\begin{equation}
J^n(\omega)=\frac{J(\omega)}{(P_{n-1}(d\lambda^0; \omega)\frac{\varphi(d\lambda^0;\omega)}{2}-Q_{n-1}(d\lambda^0; \omega))^2+ J^2(\omega)P_{n-1}^2(d\lambda^0; \omega)}
\end{equation}
$n=1,2,3,\ldots\,.$ Now we have to show that $J^n$ is the nth residual spectral density for the particle mapping Hamiltonian Eq. (\ref{marticle mapping hamiltonian eq}) i.e. the interaction described by Eq. (\ref{particle embedding eq}) when $q=0$.
First note that using Eq. (\ref{betanorm alpah n and m eq}) and (\ref{betanorm beta n and n eq}), we have
\begin{eqnarray}
\beta_0(d\vartheta_m)&=&\beta_m(d\lambda^0)\quad m=1,2,3,\ldots\,,\\
\alpha_n(d\vartheta_m)&=&\alpha_{n+m}(d\lambda^0)\quad m=1,2,3,\ldots\,,\,n=0,1,2,\ldots\,,\\
\beta_{n+1}(d\vartheta_m)&=&\beta_{n+1+m}(d\lambda^0)\quad m=1,2,3,\ldots\,,\,n=0,1,2,\ldots\,.
\end{eqnarray}
Substituting these  identities into Eq. (\ref{sequence hams in chain rep}), we see that all $\{H_m\}_{m=1}^\infty$ are equal to one another and equal to (\ref{marticle mapping hamiltonian eq}).\qed
\end{proof}
\subsection{Residual spectral densities sequence convergence}\label{convergence section}
\begin{definition}\label{sego class tdef}We say that the chain mapping for some $q$ and a particular spectral density $J(x)$, will belong to the \textup{Szeg\"o class} if
the measure $d\lambda^q(x)=M^q(x)dx$ satisfies
\begin{equation}\label{sego condition eq}
\int^{G_q(\omega_{max})}_{G_q(\omega_{min})}\frac{\ln M^q(x)\,dx}{\sqrt{\big(G_q(\omega_{max})-x\big)\big(x-G_q(\omega_{min})\big)}}>-\infty.
\end{equation}
\end{definition}
\begin{remark}
Examples of spectral densities which for any $q$ do not belong to the Szeg\"o class, are those which are gapped and those with unbounded support.
\end{remark}
\begin{theorem}\label{alpha and beta coefficients convergence}
If for some $q$ and spectral density $J(\omega)$, chain mapping  belongs to the Szeg\"o class, then the sequences $\alpha_0(q),$\,$\alpha_1(q),$\,$\alpha_2(q),\ldots$ and $\beta_0(q),$\,$\beta_1(q),$\,$\beta_2(q),\ldots$ converge to:
\begin{eqnarray}
\lim_{n \rightarrow\infty}\alpha_n(q)&=&\frac{G_q(\omega_{max})+G_q(\omega_{min})}{2},\label{lim alpha x eq}\\
\lim_{n \rightarrow\infty}\beta_n(q)&=&\frac{\big(G_q(\omega_{max})-G_q(\omega_{min})\big)^2}{16}.\label{lim beta x eq}
\end{eqnarray}
\end{theorem}
\begin{proof}
Follows from shifting the support region of the Szeg\"o theorem in \cite{Alex}. For the original theorem see \cite{original szego}. \qed
\end{proof}
\begin{corollary}\label{the other main corrollary} If for some $q$ and spectral density $J(\omega)$ the chain mapping belongs to the Szeg\"o class,
the tail of the semi-infinite chain mapping tends to a translational invariant chain. In other words,
\begin{equation}
\lim_{n\rightarrow \infty} E_{p,n}(q)=C_p(q) \quad p=1,\ldots,5\,,
\end{equation}
where $C_p\in\mathbb{R}$ are finite constants for all constant $q\in[0,1]$ and $E_{p,n}$ are defined in theorem (\ref{generalised mapping theorem}). 
\end{corollary}
\begin{proof}
Follows from theorem (\ref{alpha and beta coefficients convergence}) and Eq. (\ref{generalised enviroment mapping eqarray}).\qed
\end{proof}
\begin{definition}
The \textup{moment problem} for a measure $d\mu$ is said to be \textup{determined}, if it is uniquely determined by its moments.
\end{definition}
\begin{theorem}\label{finite interval moment problem theorem}
If a measure $d\mu$ has a finite support interval $I$, then its moment problem is determined.
\end{theorem}
\begin{proof}
See \cite{Gautschi}.
\end{proof}
\begin{theorem}\label{convergence of bounded mesures}
If for some gappless measure $d\mu(x)$ with finite support interval $I$ the limits
\begin{eqnarray}
\lim_{n \rightarrow\infty}\alpha_n(d\mu)&=&\frac{a+b}{2},\label{lim alpha ab eq}\\
\lim_{n \rightarrow\infty}\beta_n(d\mu)&=&\frac{\left(b-a\right)^2}{16},\label{lim beta ab eq}
\end{eqnarray}
exist,
then the sequence of beta normalised and normalised secondary measures generated from $d\nu$ and $d\mu$ respectively,  converge weakly to
\begin{eqnarray}
\lim_{n\rightarrow\infty}\bar\nu_n(x)&=&\frac{\sqrt{(x-a)(b-x)}}{2\pi},\label{lim beta measure}\\
\lim_{n\rightarrow\infty}\bar\mu_n(x)&=&\frac{8\sqrt{(x-a)(b-x)}}{\pi(b-a)^2}.\label{lim norm measure}
\end{eqnarray}
\end{theorem}
\begin{proof}
First we will show that the limit exists by construction (this is to say, by finding the fixed points of the sequence), then we will show that it corresponds to when the limits Eq. (\ref{lim alpha ab eq}) and (\ref{lim beta ab eq}) are accomplished.\\
Substituting Eq. (\ref{intermed reducer reducer}) into Eq. (\ref{eq for rho in temrs of mu}), we find
\begin{equation}\label{ver eq for proof of con}
\varphi(d\rho; x)=2\big[x-C_1(d\mu)\big]-\frac{\varphi(d\mu;x)\bar\rho(x)}{\bar\mu(x)}.
\end{equation}
If the limit exists, then there must be at least one solution to $\bar\rho(x)=A\bar\mu(x)$ where $d\rho(x)=\bar\rho(x)dx$ is the secondary measure associated with the gapless measure $d\mu(x)=\bar\mu(x)dx$ and $A>0$.
 From definitions (\ref{def reducer}) and (\ref{sitjes def}) we see that $\varphi(d\rho;x)=A\varphi(d\mu;x)$ and hence from Eq. (\ref{ver eq for proof of con}) we see that
\begin{equation}
\varphi(d\mu;x)=\frac{x-C_1(d\mu)}{A}.
\end{equation}
Thus substituting this into Eq. (\ref{eq for rho in temrs of mu}), and solving for $\bar\mu(x)$ we find
\begin{equation}
\bar\mu^2(x)=\frac{4A-(x-C_1(d\mu))^2}{4\pi^2A^2}.
\end{equation}
Taking into account the definition of a gapless measure, definition (\ref{meausre defn}), we see that if the limit exits, then it must be bounded. Moreover, it must belong to the interval centered at $C_1(d\mu)$ and of length $4\sqrt{A}$. Making the change of variable $x-C_1(d\mu)=2\sqrt{A}t$ gives us
\begin{eqnarray}\label{-1,1 mesure eq}
d\mu(t)=\bar\mu(t)dt=\frac{2\sqrt{1-t^2}}{\pi}dt,
\end{eqnarray}
with support interval [-1,1]. We now can check that Eq. (\ref{-1,1 mesure eq}) exists by direct substitution. Using definition (\ref{sitjes def}) we find that the Stieltjes Transform of Eq. (\ref{-1,1 mesure eq}) is $S_{\bar\mu}(z)=2[z-\sqrt{z^2-1}]$ and hence from definition (\ref{def reducer}) we find $\varphi(d\mu;x)=4x$. Thus using Eq. (\ref{eq for rho in temrs of mu}) we have $\bar\rho(x)=\bar\mu(x)/4$, hence  the limit exists. By performing the change of variable $t=(2/(b-a))y+(b+a)/(a-b)$, we shift the support of Eq. (\ref{-1,1 mesure eq}) to the general case $[a,b]$ and the measure is now given by Eq. (\ref{lim norm measure}).\\
Now we will proceed to show that if Eq. (\ref{lim alpha ab eq}) and (\ref{lim beta ab eq}) are satisfied, then the sequence of normalised secondary measures converge weakly to Eq. (\ref{lim norm measure}).
By taking the nth moment of the measures in Eq. (\ref{norm 2nd meausres eqs3}) and taking into account lemma (\ref{Rolands R eq}) for $n=0$ , we have
\begin{equation}\label{next line sub eq}
C_n(d\rho_{n+1})=\left[C_2(d\mu_n)-C_1(d\mu_n)^2\right]C_n(d\mu_{n+1})\quad n,m=0,1,2,\ldots\,.
\end{equation}
Writing Eq. (\ref{Rolands R eq}) for a sequence of measures followed by substituting in our expression for $C_n(d\rho_{n+1})$ using Eq. (\ref{next line sub eq}), we find
\begin{equation}\label{the moments recurrence eq}
\left(c_2^s-(c_1^s)^2 \right)c_n^{s+1}=c_{n+2}^s-c_1^sc_{n+1}^s-\sum_{j=0}^{n-1}\left(c_2^s-(c_1^s)^2\right)c_j^{s+1}c_{n-j}^s\quad n,s=0,1,2,\ldots\,,
\end{equation}
where $c_n^s:=C_n(d\mu_s)\quad n,s=0,1,2,\ldots$. Let us define the limit $l_n=\lim_{s\rightarrow\infty}c_n^s$. We can now draw the following conclusions
\begin{itemize}
\item[1)] Due to corollary (\ref{normalised sencond coll}), we have $C_0(d\mu_n)=1\quad n=0,1,2,\ldots\,.$ Therefore $l_0=1$.
\item[2)] From Eq. (\ref{alpah n and 0 eq}) and (\ref{the alpha def eq}) we see that
\begin{equation}\label{the alpah n 1st mo n eq}
\alpha_n(d\mu_0)=\alpha_0(d\mu_n)=\frac{C_1(d\mu_n)}{C_0(d\mu_n)}=c_1^n.
\end{equation}
Hence taking into account assumption Eq. (\ref{lim alpha ab eq}), we have $l_1=(a+b)/2$.
\item[3)] Noting the definition of $d_n$ in theorem (\ref{lemma continued frac}), theorem (\ref{the dn beta theorem}) tells us
\begin{equation}
\beta_{n+1}(d\mu_0)=c_2^n-\left(c_1^n\right)^2\quad n=0,1,2,\ldots\,.
\end{equation}
Hence taking into account assumptions Eq. (\ref{lim alpha ab eq}) and (\ref{lim beta ab eq}), and Eq. (\ref{the alpah n 1st mo n eq}) we have $l_2=(5a^2+6ab+5b^2)/16$.
\end{itemize}
We note that for any $s$, all the moments $c_n^s$ in Eq. (\ref{the moments recurrence eq}) are fully determined by the starting values $c_0^s$, $c_1^s$, and $c_2^s$, hence we conclude from the above points 1), 2) and 3), that under the assumptions Eq.  (\ref{lim alpha ab eq}) and (\ref{lim beta ab eq}),
all $l_n \quad n=0,1,2,\ldots$are finite and determined by
\begin{equation}\label{the moments recurrence  for l eq}
\left(l_2-l_1^2 \right)l_n=l_{n+2}-l_1l_{n+1}-\left(l_2-l_1^2 \right)\sum_{j=0}^{n-1}l_jl_{n-j}\quad n,s=0,1,2,\ldots\,,
\end{equation}
with starting values $l_0=1$, $l_1=(a+b)/2$, and $l_2=(5a^2+6ab+5b^2)/16$.\\ For the case of the normalised measure Eq. (\ref{lim norm measure}), we conclude that $c_n^{s}=c_n^{s+1}\quad n,s=0,1,2,\ldots$ since its secondary normalised
measure is equal to itself. Hence by denoting $m_n=c_n^s\quad n,s=0,1,2,\ldots$ we can write Eq. (\ref{the moments recurrence eq}) for this measure as
\begin{equation}\label{the moments recurrence for m eq}
\left(m_2-(m_1)^2 \right)m_n=m_{n+2}-m_1m_{n+1}-\left(m_2-(m_1)^2\right)\sum_{j=0}^{n-1}m_jm_{n-j}
\end{equation}
$n,s=0,1,2,\ldots\,. $ By direct calculation of the moments of Eq. (\ref{lim norm measure}), we find that $m_0=l_0$, $m_1=l_1$, $m_2=l_2$, and hence by comparing Eq. (\ref{the moments recurrence  for l eq}) with Eq. (\ref{the moments recurrence for m eq})
we conclude that $m_n=l_n\quad n=0,1,2,\ldots\,.$ Thus using lemma (\ref{finite interval moment problem theorem}) we conclude  Eq. (\ref{lim norm measure}) under the assumptions  Eq. (\ref{lim alpha ab eq}) and (\ref{lim beta ab eq}).
 For Eq. (\ref{lim beta measure}), we note that Eq. (\ref{nu lemma eq}) and Eq. (\ref{mu nu beta eq}) tell us $\lim_{n\rightarrow\infty}\bar\nu_n(x)=\lim_{n\rightarrow\infty}\beta_n(d\mu_0)\bar\mu_n(x)$. Hence from Eq. (\ref{lim beta ab eq}) and (\ref{lim norm measure}) we conclude Eq. (\ref{lim beta measure}).\qed
 \end{proof}
\begin{definition}
We will call \textup{terminal spectral density} $J_T(\omega)$ the spectral density to which a sequence of residual spectral densities converge to weakly if such a limit exists: $J_T(\omega)=\lim_{n\rightarrow\infty}J_n(\omega)$. \end{definition}
We are now ready to state our fourth main theorem:
\begin{theorem}\label{SD convergence thorem}
If for the particle or phonon mapping the spectral density $J(\omega)$ belongs to the Szeg\"o class, the sequence of residual spectral densities converge weakly to
the Wigner semicircle distribution \textup{\cite{wigner}}
\begin{equation}\label{wigner semicircle eq}
J_T(\omega)=\frac{\sqrt{(\omega-\omega_{min})(\omega_{max}-\omega)}}{2},
\end{equation}
and
the Rubin model spectral density \textup{\cite{weiss}}
\begin{equation}\label{rubin terminal SD eq}
J_T(\omega)=\frac{\sqrt{(\omega^2-\omega_{min}^2)(\omega_{max}^2-\omega^2)}}{2},
\end{equation}
respectively, $\omega\in[\omega_{min},\omega_{max}]$.
\end{theorem}
\begin{proof}
From Eq. (\ref{the SDs2 }) and theorem (\ref{convergence of bounded mesures}) we gather that if Eq. (\ref{lim alpha ab eq}) and (\ref{lim beta ab eq}) are satisfied, then
for the phonon case
\begin{equation}
J_T(\omega)=\frac{\sqrt{(\omega^2-a)(b-\omega^2)}}{2}.
\end{equation}
From Eq. (\ref{min max int for phonon}) we have that $a=\omega_{min}^2$ and $b=\omega_{max}^2$.
Now taking into account theorem (\ref{alpha and beta coefficients convergence}), we find Eq. (\ref{rubin terminal SD eq}).
Proceeding in a similar manner, we find Eq. (\ref{wigner semicircle eq}).\qed
\end{proof}
\section{Examples}\label{examples section}
\subsection{power law spectral densities with finite support}
The widely studied power law spectral densities are \cite{SB exp cut,weiss}
\begin{equation}\label{SD for SB no exp cut eq}
J(x)=2\pi\alpha\omega_c^{1-s}x^s,
\end{equation}
with domain $[0, \omega_c]$ and $s>-1$. Let us start by calculating the sequence of residual spectral densities for the case of the particle mapping.\\
From Eq. (\ref{the measure for the particle mapping sd eq}) and (\ref{SD for SB no exp cut eq}) we have
\begin{equation}
M^0(x)=2\alpha\omega_c^{1-s}x^s.
\end{equation}
For simplicity, we will scale out the $\omega_c$ dependency and show how to put it back again afterwards. Let us start by defining the weight function
$m^0_0(x):=\omega_cM^0(x\omega_c)=2\alpha\omega_0^2x^s$ with support interval $[0,1]$. From lemma (\ref{scaling lemma}), we find that $m^0_n(x)=M^0_n(x\omega_c)/\omega_c$ $\quad n=1,2,3,\ldots$. where $m^0_n(x)$ and $M^0_n(x)$ are the sequence of beta normalised measures generated from $m^0_0(x)$ and $M^0(x)$ respectively. Now let us define the weight function
\begin{equation}\label{tilde m 0}
\tilde m^0_0(x)=x^s
\end{equation}
with support interval $[0,1]$. Given that this new measure is proportional to $m^0_0(x)$, from lemma (\ref{scaling constant I lemma}) we conclude that it's sequence of beta normalised measures $\tilde m^0_n(x)$ $\quad n=1,2,3,\ldots\,,$ are equal. Hence we have $M^0_n(x)=\omega_c\tilde m^0_n(x/\omega_c)$ $\quad n=1,2,3,\ldots\,.$ Thus taking into account Eq. (\ref{the measure for the particle mapping sd eq}) and (\ref{def of d theta for particle SD proof eq}), we conclude
\begin{equation}\label{jn for particle sb in tilde m}
J_n(\omega)=\pi M^0_n(\omega)=\omega_c\pi \tilde m^0_n(\omega/\omega_c) \quad n=1,2,3,\ldots\,.
\end{equation}
The real polynomials orthogonal to the weight function $\tilde m^0_0(x)$ are $P_n^s(x):=P_n^{(0,s)}(2x-1)\sqrt{n+s+1} \quad n=0,1,2,\ldots\,;$ which are normalised shifted counterparts of the \textit{Jacobi polynomials} $P_n^{(\alpha,\beta)}(x) \quad n=0,1,2,\ldots\,.$ The reducer for the case $s\geq0$ is given by theorem (\ref{dev measure reducer theorem})
\begin{equation}
\varphi(d\tilde m^0;t)=2\left[\ln\left(\frac{t}{1-t}\right)+s\int_0^{1}x^{s-1}\ln\left|\frac{t-x}{t}\right|dx\right],
\end{equation}
which has analytic solutions when a particular value of $s$ is specified. For example, the first three in the sequence for the ohmic case ($s=1$) are
\begin{eqnarray}
\tilde m^0_1(x)&=&\frac{x}{2\left(\pi^2x^2+[1+x\ln(\frac{1-x}{x})]^2 \right)},\\
\tilde m^0_2(x)&=&\frac{x}{4\pi^2(2-3x)^2x^2+\left[1-6x+(4-6x)x\ln(\frac{1-x}{x}) \right]^2},\\
\tilde m^0_3(x)&=&\frac{6x}{36\pi^2x^2(3-12x+10x^2)^2+\left[30x-16+(18-72x+60x^2)(1+x\ln(\frac{1-x}{x})) \right]^2}.\nonumber
\end{eqnarray}
The chain coefficients are
calculated in \cite{Alex} to be
\begin{eqnarray}
\alpha_n(0)(s)&=&\frac{\omega_c}{2}\left(1+\frac{s^2}{(s+2n)(2+s+2n)}\right)\quad n=0,1,2,\ldots\,,\\
\sqrt{\beta_{n+1}(0)}(s)&=&\frac{\omega_c(1+n)(1+s+n)}{(s+2+2n)(3+s+2n)}\sqrt{\frac{3+s+2n}{1+s+2n}}\quad n=0,1,2,\ldots\,,\quad \\
\end{eqnarray}
 where the $(s)$ is to remind us of their $s$ dependency. System environment coupling coefficient
is \cite{Alex}
\begin{equation}
\sqrt{\beta_0(0)}=\omega_c\sqrt{\frac{2\alpha}{s+1}}.
\end{equation}
Now we will find the sequence of residual spectral densities for the case of the phonon mapping. From Eq. (\ref{M1 eq}) and (\ref{SD for SB no exp cut eq}) we have
\begin{equation}\label{the M1 finally eq}
M^1(x)=2\alpha\omega_c^{1-s}x^{s/2},
\end{equation}
with support $[0, \omega_c^2]$. Proceeding as in the particle mapping case and taking into account Eq. (\ref{the SDs2 }), we find
\begin{equation}\label{jn for phonon sb in tilde m}
J_n(\omega)=\pi M^1_n(\omega^2)=\omega_c^2\pi \tilde m^1_n(\omega^2/\omega_c^2) \quad n=1,2,3,\ldots\,,
\end{equation}
where $ \tilde m^1_n(x) \quad n=1,2,3,\ldots\,,$ are the sequence of beta normalised measures generated from
\begin{equation}\label{tilde m 1}
\tilde m^1_0(x)=x^{s/2},
\end{equation}
with support $[0,1]$. Let us denote $\tilde m^1_0(x)$ and $\tilde m^0_0(x)$ by $\tilde m^1_{0s}(x)$ and $\tilde m^0_{0s}(x)$ respectively to remind us of their $s$ dependency. By comparing eq. (\ref{tilde m 1}) with eq. (\ref{tilde m 0}), we have that $\tilde m^1_{0s}(x)=\tilde m^0_{0s/2}(x)$ for all $s>-1$ and hence  from Eq. (\ref{jn for particle sb in tilde m}) and (\ref{jn for phonon sb in tilde m}) we see that for the Spin-Boson models, there is a simple relationship between the residual spectral densities of the particle and phonon mappings for different $s$ values. For the same example as in the particle case $(s=1)$, we need to evaluate $m^0_n(x)$ $n=1,2,3,\ldots\,,$ for $s=1/2$. The first one in the sequence is
\begin{equation}
\tilde m^1_1(x)=\frac{2\sqrt{x}}{3\left(\pi^2x+(2-2\sqrt{x}\,\text{tanh}^{-1}(\sqrt{x}))^2\right)},
\end{equation}
where $\text{tanh}^{-1}(x)$ is the inverse hyperbolic tangent function. We can readily calculate the chain coefficients from the particle example. By comparing the expression for the $\alpha_n$ and $\beta_n$ coefficients for the weight functions for the particle and phonon mappings, we find
\begin{eqnarray}
\alpha_n(1)(s)&=&\omega_c\alpha_n(0)(s/2)\quad n=0,1,2,\ldots\,,\\ \sqrt{\beta_{n+1}(1)}(s)&=&\omega_c\sqrt{\beta_{n+1}(0)}(s/2)\quad\ n=0,1,2,\ldots\,, \end{eqnarray}
and system environment coupling term to be
\begin{equation}
\sqrt{\beta_0(1)}=2\omega_c\sqrt{\frac{\alpha\omega_c}{s+2}}.
\end{equation}
For both the particle and phonon mappings, it is easy to verify that the chain coefficients and sequence of residual spectral densities will converge because Eq. (\ref{sego condition eq}) is satisfied in both cases as long as $s<\infty$. We also see that the sequence of residual spectral densities calculated in the above examples converge very rapidly to this limit after about the 3rd residual spectral density.
\subsection{The power law spectral densities with exponential cut off}\label{The power law spectral densities with exponential cut offff}
The power law spectral densities with exponential cut off is \cite{SB exp cut,weiss}
\begin{equation}\label{SD for SB exp cut eq}
J(x)=2\pi\alpha\omega_c^{1-s}x^se^{-x/\omega_c},
\end{equation}
with domain $[0, \infty)$ and $s>-1$. Let us start by calculating the sequence of residual spectral densities for the case of the particle mapping.\\

From Eq. (\ref{the measure for the particle mapping sd eq}) and (\ref{SD for SB exp cut eq}) we have
\begin{equation}
M^0(x)=2\alpha\omega_c^{1-s}x^se^{-x/\omega_c}.
\end{equation}
Let us define the measure $m^0(x):=\omega_cM^0(x\omega_c)=2\alpha\omega_0^2x^se^{-x}$ with support interval $[0,\infty)$. From lemma (\ref{scaling lemma}), we find that $m^0_n(x)=M^0_n(x\omega_c)/\omega_c$ $\quad n=1,2,3,\ldots$. where $m^0_n(x)$ and $M^0_n(x)$ are the sequence of beta normalised measures generated from $m^0(x)$ and $M^0(x)$ respectively. We will now define a 3rd measure by $\tilde m^0(x)=x^se^{-x}$. We note that it is proportional to the measure $m^0(x)$ and hence lemma (\ref{scaling constant I lemma}) tells us that its sequence of beta normalised measures are equal. Thus we have the relation
\begin{equation}\label{M0 and bar m0 eq}
M^0_n(x)=\omega_c\tilde m^0_n(x/\omega_c) \quad n=1,2,3,\ldots\,.
\end{equation}
The real polynomials orthogonal to the weight function $\bar m^0(x)$ are called the \textit{associated Laguerre polynomials} $L_n^s(x) \quad n=0,1,2,\ldots\,.$
Their normalised counterparts are $P_n^s(x):=L_n^s(x)\!\,\ n!/\Gamma(n+s+1) \quad n=0,1,2,\ldots\,.$ The reducer in this case is given by theorem (\ref{dev measure reducer theorem})
\begin{equation}
\varphi(d\tilde m^0;x)=2\int_0^{+\infty}(s-t)t^{s-1}e^{-t}\ln\left|\frac{t-x}{x}\right|dt,
\end{equation}
which has analytic values when a particular value of $s$ is specified. For example, if $s$ is integer, we find
\begin{equation}
\varphi(d\tilde m^0;x)=2\left[x^se^{-x}Ei(x)-\sum_{k=0}^{k=s-1}(s-k-1)!\, x^k\right],
\end{equation}
where $Ei(x)$ is the \textit{exponential integral} function\cite{exp func}.\\
From Eq. (\ref{def of d theta for particle SD proof eq}) followed by Eq. (\ref{M0 and bar m0 eq}) we have
\begin{equation}
J_n(\omega)=\pi M^0_n(\omega)=\omega_c\pi \tilde m^0_n(\omega/\omega_c) \quad n=1,2,3,\ldots\,,
\end{equation}
and Eq. (\ref{the nu n eq}) tells us
\begin{equation}
\tilde m^0_n(x)=\frac{ x^se^{-x}}{\left(P_{n-1}^s(x)\frac{\varphi(d\tilde m^0; x)}{2}-Q^s_{n-1}(x)\right)^2+\pi^2 x^{2s}e^{-2x}P_{n-1}^s(x)^2}\quad n=1,2,3,\ldots\,.
\end{equation}
For example, the first two in the sequence for the ohmic case ($s=1$) are
\begin{eqnarray}
 \tilde m^0_1(x)&=&\frac{x e^{x}}{e^{2x}+\pi^2x^2-2x Ei(x)e^{x}+x^2 Ei(x)^2},\\
 \tilde m^0_2(x)&=&\frac{2x e^{x}}{e^{2x}(1-x)^2+\pi^2x^2(x-2)^2-2x(2-3x+x^2) Ei(x)e^{x}+x^2 (x-2)^2Ei(x)^2}.\quad\quad
\end{eqnarray}
We also have analytic expressions for the chain coefficients. From \cite{Alex}, we have that the chain coefficients are
\begin{eqnarray}
\sqrt{\beta_{n+1}(0)}&=&\omega_c\sqrt{(n+1)(n+s+1)}\quad n=0,1,2,\ldots\,,\\
\alpha_n(0)&=&\omega_c(2n+1+s) \quad n=0,1,2,\ldots\,,
\end{eqnarray}
with the system-environment coupling coefficient given by
\begin{equation}
\sqrt{\beta_0(0)}=\omega_c\sqrt{2\alpha\Gamma(s+1)}.
\end{equation}
Similarly, we can calculate the residual spectral densities and chain coefficients for the phonon mapping case.\\
Because the support interval is infinite, the Spin-Boson model with exponential cut off does not belong to the Szeg\"o class for either the particle mapping nor the phonon mapping as can be easily verified from Eq. (\ref{sego condition eq}). This is reflected in the example above as the sequence of residual spectral densities do not converge\footnote{However, we do note that the ratio $\alpha_n(q)/\sqrt{\beta_{n+1}(q)}$  for $q=0$ and $q=1$ does converge. This suggests that there is a universal asymptotic expansion for large $n$ for the $n$th residual spectral density.}. \\
\begin{remark}
We note that for spectral densities for which the corresponding orthogonal polynomials are unknown analytically, we can easily calculate their coefficients using very stable numerical algorithms. See \cite{Javi} or \cite{Alex} and references herein for more details.
\end{remark}

\section{Summary and conclusion}
By developing the method of \cite{Alex,Javi}, we have established a general formalism for mapping an open quantum system of arbitrary spectral density of the form Eq. (\ref{the 1st eq of genre theorem}) onto chain representations. The different versions of chain mappings are generated by choosing particular values of 4 real constants.  This has also provided a very general connection between the theory of open quantum systems and Jacobi operator theory as the semi-infinite chains can be written in terms of Jacobi matrices as can be seen in corollary (\ref{chain in jacobi for conc}). There has been a wealth of research into the properties of Jacobi operator theory \cite{Gerard}, and hence this opens up the theorems of this field to the possibility of being applied to the theory of open quantum systems. Likewise, the new theorem regarding Jacobi matrices (\ref{jacobi matrix theorem}) developed here, could turn out to be useful in the field of Jacobi operator theory.\\
There were two previously known exact chain mappings; the one of \cite{Alex} which is the same as the particle mapping defined here (definition \ref{particle case def}), and the mapping by \cite{Bassano}. We show that the phonon mapping derived in this article (see definition \ref{phonon case def}) has a wider range of validity (remark \ref{gaps}) than the mapping of \cite{Bassano} and prove that they are formally equivalent in the range of validity of \cite{Bassano}; see theorem (\ref{the partial sd Bassano et all equiv theorem}).\\
The concept of embedding degrees of freedom into the system has been around for some time \cite{pseudomodes,Bassano}, however, we see that in chain representations, there is a natural way of shifting the system-environment boundary, that is to say, there is a natural and systematic way of embedding degrees of freedom into the system one by one (or all in one go). To solve quantitatively this problem, we  have to first embark on finding the solution to an old problem in mathematics; an analytical solution to the sequence of secondary measures in terms of the initial measure, it's associated orthogonal polynomials and reducer; see theorem (\ref{1st main theorem}). Not only does this provide a means to find analytical expressions for the sequence of spectral densities corresponding to the new system-environment interaction after embedding environmental degrees of freedom into the system for the particle and phonon mappings in the gapless spectral densities case (theorem (\ref{theorem for the sequence of sd in the particle m case}) and corollary (\ref{the phonon mapping corollary rof psd}) respectively), but it provides physical meaning to this abstract mathematical construct.\\
Using convergence theorems of Szeg\"o and deriving the fixed point in the sequence of secondary measures; we have combined these results to obtain a convergence theorem of the sequence of residual spectral densities for the particle and phonon cases, theorem (\ref{SD convergence thorem})  (or equivalently, the sequence of secondary measures, theorem (\ref{convergence of bounded mesures})). What is more, because the criterion for convergence (definition \ref{sego class tdef}) is solely in terms of the initial spectral density once the desired chain mapping is chosen, the convergence criterion is readily applicable to a particular problem. Furthermore, we see that any unbounded spectral density will not satisfy this criterion. This is reflected in the sequence generated in the examples section for the case of the family of power-law spectral densities with exponential cut off as the sequence does not converge, section \ref{The power law spectral densities with exponential cut offff}.\\
We give two examples where we can find explicit analytical expressions for the chain coefficients and the sequence of residual spectral densities. These are the family of spectral densities used in the Spin-Boson model, which have spectral densities of the form $x^s$ with finite support $I=[0,\omega_c]$ and $x^se^{x/\omega_c}$ with semi-infinite support $I=[0,+\infty)$, where $s>-1$. Furthermore, we show how the residual spectral densities of the mapping for both phonon and particle cases are related for different families of the mapping as demonstrated by the identity $\tilde m^1_{0s}(x)=\tilde m^0_{0s/2}(x)$ and Eq. (\ref{jn for particle sb in tilde m}) and (\ref{SD for SB exp cut eq}).\\
We note that when this convergence of the embeddings is achieved (i.e. inequality (\ref{sego condition eq}) is satisfied), the part of the chain corresponding to the new environment has a very universal property: all its couplings and frequencies become constant (Corollary \ref{alpha and beta coefficients convergence}). This is to say, they are translationally invariant. This suggests a universal way of simulating the environment as all the characteristic features of the environment are now embedded into the system, as discussed for the mapping of \cite{Alex} in \cite{bookchap}. What is not so clear however, is what is the most effective chain mapping for simulating the environment. The general formalism of chain mapping for open quantum systems developed here, paves the way to answering these questions.

\section{Appendix}
\subsection{Proof that the relation between the Stieltjes transforms of two consecutive measures is invalid for gapped measures}\label{break sidt}
In the case of a gapless interval $I=[a,b],$ the relation between the Stieltjes transforms of subsequent measures, Eq (\ref{eq sit rel}): $S_\rho(z)=z-C_1(d\mu)-1/S_\mu(z),\, z\in \mathbb{C}-I$ is well defined because $S_\mu(z)$ does not vanish on $\mathbb{C}-I$:
\begin{proof}
Let $z=x+iy\quad x,y\in \mathbb{R}$;
\begin{equation}
S_\mu(z)=\int_a^b\frac{\mu(t)dt}{x+iy-t}=\int_a^b\frac{(x-t)\mu(t)dt}{(x-t)^2+y^2}-iy\int_a^b\frac{\mu(t)dt}{(x-t)^2+y^2}.
\end{equation}
Hence for $S_\mu(z)=0$ we need $y=0,$ thus $z=x \in \mathbb{C}-I,$ however, under these conditions $(x-t)\mu(t)$ cannot change sign in $t \in I$, therefore $S_\mu(z) \neq 0$. \qed
\end{proof}
In the case of a gapped interval $I=[a,b]\cup [c,d]$, $S_\mu(z)$ the Stieltjes transformation is defined by
\begin{equation}
S_\mu(z)=\int_a^b\frac{\mu(t)dt}{z-t}+\int_c^d\frac{\mu(t)dt}{z-t}.
\end{equation}
Unfortunately, this expression vanishes on a point in $[b,c]$
\begin{proof}
Taking real and imaginary parts of $S_\mu(z)$ followed by setting the imaginary part to zero, we find
\begin{equation}
S_\mu(x)=\int_a^b\frac{(x-t)\mu(t)dt}{(x-t)^2}+\int_c^d\frac{(x-t)\mu(t)dt}{(x-t)^2},
\end{equation}
where $z=x \in \mathbb{R}-I$. However, the numerators of the first and second integrals have a different sign for $x \in [b,c]$. Thus due to the continuity of $S_\mu(x)$ in $]b,c[$ and the existence of the limits $\lim_{x\rightarrow b}S_\mu(x)=+\infty$, $\lim_{x\rightarrow c}S_\mu(x)=-\infty$, there exists a point $z_0$ in $\mathbb{C}-I$ such that $S_\mu(z_0)=0$. \qed
\end{proof}
Consequently, one cannot define a secondary measure as in definition (\ref{def dary measure}) which satisfy theorem (\ref{theorem stiejes}), through this relation because the Stieltjes transform must be holomorphic.
\subsection{Methods for calculating the reducer}\label{Methods for calculating the reducer}
\begin{theorem}
If for a measure $d\mu(x)=\bar \mu(x)dx$, $\bar \mu(x)$ satisfies a Lipschitz condition over its support interval $[a,b]$, then
\begin{equation}\label{the lipships version of reducer eq}
\varphi(d\mu;x)=2\bar\mu(x)\ln\left(\frac{x-a}{b-x}\right)-2\int_a^b\frac{\bar\mu(t)-\bar\mu(x)}{t-x}dt.
\end{equation}
\end{theorem}
\begin{proof}
From definitions (\ref{def reducer}) and (\ref{sitjes def}) we have
\begin{equation}
\varphi(d\mu;x)=\lim_{\epsilon\rightarrow 0^{\!+}}2\int_a^b\frac{(x-t)\bar \mu(t)dt}{(x-t)^2+\epsilon^2}.
\end{equation}
Writing this as
\begin{eqnarray}
\varphi(d\mu;x)&=&\lim_{\epsilon\rightarrow 0^{\!+}}2\int_a^b\frac{(x-t)\bar \mu(x)}{(x-t)^2+\epsilon^2}dt+\lim_{\epsilon\rightarrow 0^{\!+}}2\int_a^b\frac{(x-t)( \bar\mu(t)-\bar\mu(x))}{(x-t)^2+\epsilon^2}dt\quad \\
&=&\lim_{\epsilon\rightarrow 0^{\!+}}2\int_a^b\frac{(x-t)\bar \mu(x)}{(x-t)^2+\epsilon^2}dt-2\int_a^b\frac{\bar\mu(t)-\bar\mu(x)}{t-x}dt\label{second eq in long phi eq}\\
&&-\lim_{\epsilon\rightarrow 0^{\!+}}\int_a^b\left(\frac{2\epsilon^2}{(x-t)^2+\epsilon^2} \right)\frac{\bar\mu(t)-\bar\mu(x)}{x-t}dt.\label{last eq in long phi eq}
\end{eqnarray}
\begin{itemize}
\item[1)]The first integral in line (\ref{second eq in long phi eq}) reduces to
\begin{eqnarray}
\lim_{\epsilon\rightarrow 0^{\!+}}2\int_a^b\frac{(x-t)\bar \mu(x)}{(x-t)^2+\epsilon^2}dt&=&\lim_{\epsilon\rightarrow 0^{\!+}}\bar\mu(x)\left[ \ln((a-x)^2+\epsilon^2)-\ln((b-x)^2+\epsilon^2)\right]\quad \quad\\
&=&2\bar\mu(x)\ln\left(\frac{x-a}{b-x}\right).
\end{eqnarray}
\item[2)] Now imposing the Lipschitz condition on $\bar\mu(x)$:  $|\bar\mu(x)-\bar\mu(t)|\leq K|t-x|$ for some $K$ over interval $[a,b]$, the absolute value of the expression on line (\ref{last eq in long phi eq}) reduces to
\begin{eqnarray}
&&\lim_{\epsilon\rightarrow 0^{\!+}}\int_a^b\left(\frac{2\epsilon^2}{(x-t)^2+\epsilon^2} \right)\frac{|\bar\mu(t)-\bar\mu(x)|}{|t-x|}dt\\
&\leq&\lim_{\epsilon\rightarrow 0^{\!+}}\int_a^b\frac{2K\epsilon^2}{(x-t)^2+\epsilon^2}dt\\
&=&\lim_{\epsilon\rightarrow 0^{\!+}}2K\epsilon\left[\arctan\left(\frac{b-x}{\epsilon}\right)-\arctan\left(\frac{a-x}{\epsilon}\right)\right]\\
&=&0.
\end{eqnarray}
Hence expression on line (\ref{last eq in long phi eq}) vanishes.
\end{itemize}\qed
\end{proof}
\begin{lemma}\label{the lipships c1 lemma}
If $f(x)$ possesses a bounded continuous first derivative  on its domain, then it also satisfies a Lipschitz condition on its domain.
\end{lemma}
\begin{proof}
See \cite{class 1 lipships rel}.
\end{proof}
\begin{theorem}\label{dev measure reducer theorem}
If for a measure $d\mu(x)=\bar \mu(x)dx$, $\bar \mu(x)$ possesses a bounded continuous first derivative  on its support interval  $I$ and $\bar \mu(x)$ can be evaluated on the limits $a,b$, then
\begin{equation}
\varphi(d\mu;x)=2\left[\bar\mu(a)\ln\left|\frac{x-a}{x}\right|+\bar\mu(b)\ln\left|\frac{x}{b-x}\right|+\int_a^b\bar\mu'(t)\ln\left|\frac{t-x}{x}\right|dt\right],
\end{equation}
where $\bar \mu'(t)$ is the first derivative of $\bar \mu(t)$.
\end{theorem}
\begin{proof}
Follows from integrating Eq. (\ref{the lipships version of reducer eq}) by parts and
applying lemma (\ref{the lipships c1 lemma}).\qed
\end{proof}
\subsection{Scaling properties of the sequence of beta normalised measures}
\begin{lemma}\label{scaling lemma}
Suppose we have two measures $d\nu^1(x)=\bar\nu^1(x)dx$ and $d\nu^2(x)=\bar\nu^2(x)dx$ with support intervals $I^1,\,I^2$ bounded by $\lambda a,\lambda b$ and $ a, b$ respectively. If they are related by
\begin{equation}
\bar\nu^1(x)=\frac{\bar\nu^2(x/\lambda)}{\lambda},\quad\  \lambda>0,\label{scaling eq}
\end{equation}
then their sequence of beta normalised measures have the relation
\begin{equation}\label{scaling lemma result eq}
\bar\nu^1_n(x)=\lambda\bar\nu^2_n(x/\lambda)\quad n=1,2,3,\ldots\,.
\end{equation}
\end{lemma}
\begin{proof}
Using the definition of inner product, definition (\ref{def inner product}), and relation eq. (\ref{scaling eq}), we find
\begin{equation}
\langle f,g\rangle_{\bar\nu^2}=\int_{a}^{b}f(\lambda t)g(\lambda t)d\nu^1(t).
\end{equation}
Hence we conclude
\begin{equation}\label{the polys nu1nu2 rel eq}
P_n(d\nu^1;x)=P_n(d\nu^2;x/\lambda)\quad n=0,1,2,\ldots\,.
\end{equation}
Using this relation and definition (\ref{defn of 2nd poly}), we find
\begin{equation}\label{the 2nd polys nu1nu2 rel eq}
Q_n(d\nu^1;x)=\frac{Q_n(d\nu^2;x/\lambda)}{\lambda}\quad n=0,1,2,\ldots\,.
\end{equation}
Similarly, we have from definition (\ref{def reducer})
\begin{equation}\label{the reducer nu1nu2 rel eq}
\varphi(d\nu^1;x)=\frac{\varphi(d\nu^1;x/\lambda)}{\lambda}.
\end{equation}
Finally, substituting Eq. (\ref{the polys nu1nu2 rel eq}), (\ref{the 2nd polys nu1nu2 rel eq}), and (\ref{the reducer nu1nu2 rel eq}) into definition (\ref{nu sequence defn}) we arrive at Eq. (\ref{scaling lemma result eq}).\qed
\end{proof}
\begin{lemma}\label{scaling constant I lemma}
The sequence of beta normalised measures generated from a measure $d\nu(x)=\bar\nu(x)dx$ are invariant under the mapping $d\nu(x)\rightarrow \lambda d\nu(x)$, $\lambda>0$ while keeping the support interval $I$ unchanged.
\end{lemma}
\begin{proof}
Follows with a similar line of reasoning as for lemma (\ref{scaling lemma}).\qed
\end{proof}

\begin{acknowledgement}
We acknowledge Michael Keyl, Christian Gerard, Sebastian Karl Egger for discussions regarding mathematical technicalities, Nicolas Neumann for pointing out typos and Ramil Nigmatullin for help with latex formatting.
This work was supported by the EPSRC CDT on Controlled Quantum Dynamics, the EU STREP project CORNER and the EU Integrating project Q-ESSENCE as well as the Alexander von Humboldt Foundation.
\end{acknowledgement}


\end{document}